\documentclass[12pt]{article}
\usepackage{graphicx}
\usepackage{amsmath}

\newtheorem{theorem}{Theorem}[section]
\newtheorem{definition}[theorem]{Definition}
\newtheorem{lemma}[theorem]{Lemma}

\newtheorem{proposition}[theorem]{Proposition}
\newtheorem{corollary}[theorem]{Corollary}
\newtheorem{remark}[theorem]{Remark}
\newenvironment{proof}[1][Proof]{\textsc{#1.\ }} {\ \rule{0.5em}{0.5em}}
\numberwithin{equation}{section}
\numberwithin{figure}{section}
\begin{document}

\title{\textbf{An existence theorem for the Cauchy problem on the light-cone
for the vacuum Einstein equations with near-round analytic data} }
\author{Yvonne Choquet-Bruhat \\
Acad\'emie des Sciences, Paris \and Piotr T. Chru\'{s}ciel \\
Universit\"{a}t Wien \and Jos\'e M. Mart\'in-Garc\'ia \\
Institut d'Astrophysique de Paris, and \\
Laboratoire Univers et Th\'{e}ories, Meudon}
\maketitle

\bigskip

\section{Introduction}

In the paper~\cite{CCM2}, which will in what follows be referred to as I, we
studied the Cauchy problem for the Einstein equations with data on a
characteristic cone $C_{O}$. We used the tensorial splitting of the Ricci
tensor of a Lorentzian metric $g$ on a manifold $V$ as the sum of a
quasidiagonal hyperbolic system acting on $g$ and a linear first order
operator acting on a vector $H,$ called the wave-gauge vector. The vector $H$
vanishes if $g$ is in wave gauge; that is, if the identity map is a wave map
from $(V,g)$ onto $(V,\hat{g})$, with $\hat{g}$ some given metric, which we
have chosen to be Minkowski. The data needed for the reduced PDEs is the
trace, which we denote by $\bar{g}$, of $g$ on $C_{O}$. However, because of
the constraints, the intrinsic, geometric, data is a degenerate quadratic
form $\tilde{g}$ on $C_{O}$. Given $\tilde{g}$, the trace $\bar{g}$ is
determined through a hierarchical system of ordinary differential equations%
\footnote{%
For previous writing of these equations in the case of two intersecting
surfaces in four-dimensional spacetime see Rendall~\cite{RendallCIVP} and
Damour-Schmidt~\cite{DamourSchmidt}.} along the rays of $C_{O}$, deduced
from the contraction of the Einstein tensor with a tangent to the rays,
which we have written explicitly and solved. We have called these equations
the wave map gauge constraints and shown that they are necessary and
sufficient conditions for the solutions of the hyperbolic system to satisfy
the full Einstein equations. We have also proved local geometric uniqueness
of a solution $g$ of the vacuum Einstein equations inducing a given $\tilde{g%
}$ (for details see I). Further references to previous works on the problem
at hand can be found in I.

Existence theorems known for quasilinear wave equations with data on a
characteristic cone give also existence theorems for the Einstein equations,
if the initial data is Minkowski in a neighbourhood of the vertex. For more
general data problems arise due to the apparent discrepancy between the
functional requirements on the characteristic data of the hyperbolic system
and the properties of the solutions of the constraints, due to the
singularity of the cone $C_{O}$ at its vertex $O$. The aim of this work is
to  make progress towards resolving this issue, and provide a sufficient
condition for the validity of an existence theorem in a neighbourhood of $O$
under conditions alternative to the fast-decay conditions of~\cite{CCM3}.
More precisely, we prove that analytic initial data arising from a metric
satisfying (\ref{4.3})-(\ref{4.4}) together with the \textquotedblleft
near-roundness" condition of Definition~\ref{Def} lead to a solution of the
vacuum Einstein equations to the future of the light-cone.

\section{Cauchy problem on a characteristic cone for quasilinear wave
equations}

The reduced Einstein equations in wave-map gauge and Minkowski target are a
quasi-diagonal, quasi-linear second order system for a set $v$ of scalar
functions $v^{I}$, $I=1,\ldots ,N$, on $\mathbf{R}^{n+1}$ of the form
\begin{equation}
A^{\lambda \mu }(y,v)D_{\lambda \mu }^{2}v+f(y,v,Dv)=0,\quad y=(y^{\lambda
})\in \mathbf{R}^{n+1},\quad n\geq 2,\;\quad f=(f^{I})  \label{2.1}
\end{equation}
If the target is the Minkowski metric and takes in the coordinates $%
y^{\alpha }$ the canonical form
\begin{equation}
\eta \equiv -(dy^{0})^{2}+\sum_{i=1}^{n}(dy^{i})^{2},  \label{2.2}
\end{equation}
then,
\begin{equation}
Dv=(\frac{\partial v^{I}}{\partial y^{\lambda }}),\qquad D_{\lambda \mu
}^{2}v=(\frac{\partial ^{2}v^{I}}{\partial y^{\lambda }\partial y^{\mu }}%
),\qquad \lambda ,\mu =0,1,\ldots ,n  \label{2.3}
\end{equation}
\emph{We will underline components in these $y^{\alpha}$ coordinates.}

In the case of the Einstein equations the functions $A^{\lambda\mu}\equiv
g^{\lambda\mu}$ do not depend directly on $y$, they are analytic in $v$ in
an open set $W\subset \mathbf{R}^{N}$. For $v\in W$ the quadratic form $%
g^{\lambda \mu }$ is of Lorentzian signature. The functions $f^{I}$ are
analytic in $v\in W$ and $Dv\in \mathbf{R}^{(n+1)N}$, they do not depend
directly on $y$ in vacuum.

The characteristic cone $C_O$ of vertex $O$ for a Lorentzian metric $g$ is
the set covered by future directed null geodesics issued from $O$. We choose
coordinates $y^{\alpha }$ such that the coordinates of $O$ are $y^{\alpha}=0$
and the components $A^{\lambda\mu}(0,0)$ take the diagonal Minkowskian
values, $(-1,1,\ldots ,1)$. If $v$ is $C^{1,1}$ in a neighbourhood $U$ of $O$
and takes its values in $W$ there is an eventually smaller neighbourhood of $%
O$, still denoted $U$, such that $C_{O}\cap U$ is an $n$ dimensional
manifold, differentiable except at $O$, and there exist in $U$ coordinates $%
y:=(y^{\alpha })\equiv (y^{0}$, $y^{i}$, $i=1,\ldots ,n)$ in which $C_{O}$
is represented by the equation of a Minkowskian cone with vertex $O$,
\begin{equation}
C_{O}:=\{r-y^{0}=0\},\quad r:=\{\sum (y^{i})^{2}\}^{\frac{1}{2}},
\label{2.4}
\end{equation}
and the null rays of $C_{O}$ represented by the generators of the
Minkowskian cone, i.e. tangent to the vector $\ell $ with components $%
\underline{\ell^{0}}=1$, $\underline{\ell^{i}}=r^{-1}y^{i}$. Inspired by
this result and following previous authors we will set the Cauchy problem
for the equations (\ref{2.1}) on a characteristic cone as the search of a
solution which takes given values on a manifold represented by an equation
of the form (\ref{2.4}), that is a set $v$ such that
\begin{equation}
\bar{v}=\varphi ,  \label{2.5}
\end{equation}
where \emph{overlining means restriction to} $C_{O}$. The function $\varphi$
takes its values in $W$ and is such that $\ell $ is a null vector for $\bar{A%
}$, i.e.when $\bar{A}\equiv \bar{g}$
\begin{equation}
\underline{\ell ^{\mu }\ell ^{\nu }\bar{g}_{\mu \nu }}=\underline{\bar{g}%
_{00}}+2r^{-1}y^{i}\underline{\bar{g}_{0i}}+r^{-2}y^{i}y^{j}\underline{\bar{g%
}_{ij}}=0.  \label{2.6}
\end{equation}

We use the following notations:
\begin{align*}
C_{O}^{T} & :=C_{O}\cap \{0\leq t:=y^{0}\leq T\}\;. \\
Y_{O} & :=\{y^{0}>r\}\;,\quad \text{the interior of\ }C_{O}\;, \\
Y_{O}^{T} & :=Y_{O}\cap \{0\leq y^{0}\leq T\}\;.
\end{align*}
and we set
\begin{align*}
\Sigma_{\tau}& :=C_{O}\cap \{y^{0}=\tau \}\;,\quad \text{diffeomorphic to\ }%
S^{n-1}\;, \\
S_{\tau}& :=Y_{O}\cap \{y^{0}=\tau \}\;,\quad \text{diffeomorphic to the
ball\ }B^{n-1}\;.
\end{align*}

We recall the following theorem, which applies in particular to the reduced
Einstein equations

\begin{theorem}
\label{ThDossanew1} Consider the problem (\ref{2.1}, \ref{2.5}). Suppose
that:

1. There is an open set $U\times W\subset \mathbf{R}^{n+1}\times \mathbf{R}%
^{N}$, $Y_{O}^{T}\subset U$ where the functions $g^{\lambda\mu}$ are smooth
in $y$ and $v$. The function $f$ is smooth\footnote{%
Smooth means $C^{m}$, with $m$ some integer depending on the problem at hand
and the considered function. In particular $C^{\infty}$ and $C^{\omega}$
(real analytic functions) are smooth.} in $y\in U$ and $v\in W$ and in $%
Dv\in \mathbf{R}^{(n+1)N}$.

2. For $(y,v)\in U\times W$ the quadratic form $g^{\lambda\mu}$ has
Lorentzian signature; it takes the Minkowskian values for $y=0$ and $v=0$.
It holds that $\varphi(O)=0$

3. a. The function $\varphi$ takes its values in $W$. The cone $C_{O}^{T}$
is null for the metric $g^{\lambda\mu}(y,\varphi)$.

b. $\varphi $ is the trace on $C_{O}^{T}$ of a smooth function in $U$.

Then there is a number $0<T_{0}\leq T<+\infty $ such that the problem (\ref%
{2.1}, \ref{2.5}) has one and only one solution $v$ in $Y_{O}^{T_{0}}$ which
can be extended by continuity to a smooth function defined on a
neighbourhood of the origin in ${\mathbf{R}}^{n+1}$.

If $\varphi $ is small enough in appropriate norms, then $T_{0}=T$.
\end{theorem}

\section{Null adapted coordinates}

It has been shown\footnote{%
See I and references therein.} that the constraints are easier
to solve in coordinates $x^{\alpha }$ adapted to the null
structure of $C_{O}$, defined by
\begin{equation}
x^{0}=r-y^{0},\qquad x^{1}=r\text{ \ \ and }x^{A}=\mu ^{A}(r^{-1}y^{i})
\label{3.1}
\end{equation}
$A=2,...n$, local coordinates on the sphere $S^{n-1}$, or angular polar
coordinates. Conversely
\begin{equation*}
y^{0}=x^{1}-x^{0},\qquad y^{i}=r\Theta ^{i}(x^{A})\qquad \text{with}\quad
\sum_{i=1}^{n}\Theta ^{i}(x^{A})^{2}=1.
\end{equation*}
In the $x$ coordinates the Minkowski metric (\ref{2.2}) reads
\begin{equation}
\eta \equiv -(dx^{0})^{2}+2dx^{0}dx^{1}+(x^{1})^{2}s_{n-1},  \label{3.2}
\end{equation}
with
\begin{equation*}
s_{n-1}:=s_{AB}dx^{A}dx^{B}\;,\quad \text{the metric of the round sphere\ }%
S^{n-1}.
\end{equation*}
Recall that in these coordinates the non zero Christoffel symbols of the
Minkowki metric are, with $S_{BC}^{A}$ the Christoffel symbols of the metric
$s$,
\begin{equation}
\hat{\Gamma}_{1A}^{B}\equiv \frac{1}{x^{1}}\delta _{A}^{B}\;,\qquad \hat{%
\Gamma}_{AC}^{B}\equiv S_{AC}^{B}\;,\text{ \ \ }\hat{\Gamma}_{AB}^{0}\equiv
-x^{1}s_{AB}\;,\qquad \hat{\Gamma}_{AB}^{1}\equiv -x^{1}s_{AB}\;.
\label{3.3}
\end{equation}
In the general case, the null geodesics issued from $O$ have still equation $%
x^{0}=0$, $x^{A}=$constant, so that $\ell :=\frac{\partial }{\partial x^{1}}$
is tangent to those geodesics. The trace $\bar{g}$ on $C_{O}$ of the
spacetime metric $g$ that we are going to construct is such that $\overline{g%
}_{11}=0$ and $\bar{g}_{1A}=0;$ we use the notation
\begin{equation}
\bar{g}\equiv \bar{g}_{00}(dx^{0})^{2}+2\nu _{0}dx^{0}dx^{1}+2\nu
_{A}dx^{0}dx^{A}+\bar{g}_{AB}dx^{A}dx^{B},  \label{3.4}
\end{equation}
We emphasize that our assumption that $\bar{g}$ is given by (\ref{3.4}) is
no geometric restriction for a metric $g$ to have such a trace on a null
cone $x^{0}=0$.

The Lorentzian metric $g$ induces on $C_{O}$ a degenerate quadratic form $%
\tilde{g}$ which reads in coordinates $x^{1}$, $x^{A}$
\begin{equation}
\tilde{g}\equiv \tilde{g}_{AB}dx^{A}dx^{B},  \label{3.5}
\end{equation}
i.e.\ $\tilde{g}_{11}\equiv \tilde{g}_{1A}\equiv 0$ while $\tilde{g}%
_{AB}dx^{A}dx^{B}\equiv $ $\bar{g}_{AB}dx^{A}dx^{B}$ is an
$x^{1}$-dependent Riemannian metric on $S^{n-1}$induced on each
$\Sigma _{t}$ by $\tilde{g}$, we denote it by
$\tilde{g}_{\Sigma }$. While $\tilde{g}$ is intrinsically
defined, it is not so for $\bar{g}_{00}$, $\nu _{0}$, $\nu
_{A}$, they are gauge-dependent quantities.

Note that $\tilde{g}$ \ has a more complicated expression in coordinates $%
y^{i}$ on $C_{O}$. Since the inclusion mapping of $C_{O}$ in the coordinates
$y^{\alpha }$ is $y^{0}=r$ hence $\frac{\partial y^{0}}{\partial y^{i}}=%
\frac{y^{i}}{r}$, it holds that
\begin{equation}
\tilde{g}\equiv \underline{\tilde{g}_{ij}}dy^{i}dy^{j},\text{ \ with \ }%
\underline{\tilde{g}_{ij}}\equiv r^{-2}y^{i}y^{j}\underline{\bar{g}_{00}}%
+r^{-1}(y^{j}\underline{\bar{g}_{0i}}+y^{i}\underline{\bar{g}_{0j}})+%
\underline{\bar{g}_{ij}}.  \label{3.6}
\end{equation}

For Theorem~\ref{ThDossanew1} to apply to the wave-gauge
reduced Einstein equations, the components of the initial data
in the $y$ coordinates must be the trace on $C_{O}$ of smooth
spacetime functions. The solution of the reduced equations
satisfy the full Einstein equations if and only if these
initial data satisfy the wave-map gauge constraints. We have
constructed in I these data as solutions of ODE in adapted null
$x$ coordinates, which are admissible coordinates for
$\mathbf{R}^{n+1}$ only for $r>0$. The change of coordinates
from $x$ to $y$, smooth for $r>0$, is
recalled below; the components of a spacetime tensor $T$ in the coordinates $%
x$ are denoted $T_{\alpha \beta }$ while in the coordinates $y$ they are
denoted $\underline{T_{\alpha \beta }}$.

\begin{lemma}
\label{lemma3.1} It holds that:
\begin{equation*}
T_{00}\equiv \underline{T_{00}},\quad T_{11}\equiv \underline{T_{00}}+2\frac{%
y^{i}}{r}\underline{T_{0i}}+\frac{y^{i}}{r}\frac{y^{j}}{r}\underline{T_{ij}}%
,\quad T_{01}\equiv -(\underline{T_{00}}+\underline{T_{0i}}\Theta ^{i}),
\end{equation*}
\begin{equation*}
T_{0A}\equiv -r\frac{\partial \Theta ^{i}}{\partial x^{A}}\underline{T_{0i}}%
,\quad T_{1A}\equiv r\frac{\partial \Theta ^{i}}{\partial x^{A}}(\underline{%
T_{0i}}+\Theta ^{j}\underline{T_{ij}}),\quad T_{AB}\equiv \underline{T_{ij}}%
r^{2}\frac{\partial \Theta ^{i}}{\partial x^{A}}\frac{\partial \Theta ^{j}}{%
\partial x^{B}}\;.
\end{equation*}
Conversely, if $T_{1A}\equiv T_{11}\equiv 0$
\begin{equation*}
\underline{T_{00}}\equiv T_{00},\quad \underline{T_{0i}}\equiv
-(T_{00}+T_{01})r^{-1}y^{i}-T_{0A}\frac{\partial x^{A}}{\partial y^{i}}\;,
\end{equation*}
\begin{equation*}
\underline{T_{ij}}=(T_{00}+2T_{01})r^{-2}y^{i}y^{j}+T_{0A}r^{-1}(y^{i}\frac{%
\partial x^{A}}{\partial y^{j}}+y^{j}\frac{\partial x^{A}}{\partial y^{i}}%
)+T_{AB}\frac{\partial x^{A}}{\partial y^{i}}\frac{\partial x^{B}}{\partial
y^{j}}.
\end{equation*}
\end{lemma}

In the following we shall often abbreviate partial derivatives as follows%
\begin{equation*}
\partial _{0}\equiv \frac{\partial }{\partial x^{0}},\quad \partial
_{1}\equiv \frac{\partial }{\partial x^{1}},\quad \partial _{A}\equiv \frac{%
\partial }{\partial x^{A}},
\end{equation*}
\begin{equation*}
\underline{\partial _{0}}\equiv \frac{\partial }{\partial y^{0}},\quad
\underline{\partial _{i}}\equiv \frac{\partial }{\partial y^{i}}.
\end{equation*}
\bigskip

\section{Characteristic data}

\subsection{Basic (intrinsic) characteristic data}

The basic data on a characteristic cone for the Einstein equation is a
degenerate quadratic form. \textbf{\ }We will define this data as the
degenerate quadratic form $\tilde{C}$ induced by a given Lorentzian
spacetime metric $C$ which admits this cone as a null cone. {We denote as
before by }$C_{O}$ the manifold $x^{0}\equiv r-y^{0}=0$. If we take the $%
y^{i}$ as coordinates on $C_{O}$ it holds that, see (\ref{3.6}),
\begin{equation}
\tilde{C}\equiv \underline{\tilde{C}_{ij}}dy^{i}dy^{j}\text{ \ \ with \ \ }%
\underline{\tilde{C}_{ij}}\equiv \underline{\bar{C}_{ij}}+r^{-1}(y^{j}%
\underline{\bar{C}_{i0}}+y^{i}\underline{\bar{C}_{j0}})+r^{-2}y^{i}y^{j}%
\underline{\bar{C}_{00}}\;,  \label{4.1}
\end{equation}
where $\underline{\bar{C}_{\alpha \beta }}$ are the components in the $y$
coordinates of the trace $\bar{C}$ of $C$ on $C_{O}$ (not to be mistaken
with the induced quadratic form $\tilde{C})$. Assuming that $C_{O}$ is a
null cone for $C$ with generators $\underline{\bar{\ell}^{0}}=1$, $%
\underline{\bar{\ell}^{i}}=\frac{y^{i}}{r}$, (\ref{2.6}) implies that the
quadratic form $\tilde{C}$ is degenerate,
\begin{equation}
\frac{y^{i}}{r}\frac{y^{j}}{r}\underline{\tilde{C}_{ij}}\equiv \underline{%
\bar{C}_{00}}+2\frac{y^{i}}{r}\underline{\bar{C}_{0i}}+\frac{y^{i}}{r}\frac{%
y^{j}}{r}\underline{\bar{C}_{ij}}=0.  \label{4.2}
\end{equation}

\begin{lemma}
Two spacetime metrics $C$ and $C^{\prime }$ with components linked by
\begin{equation*}
\underline{\bar{C}_{ij}}^{\prime }:=\underline{\bar{C}_{ij}}%
+r^{-1}(a_{i}y_{j}+a_{j}y_{i})+r^{-2}\alpha y^{i}y^{j},\text{ \ }\underline{%
\bar{C}_{i0}}^{\prime }:=\underline{\bar{C}_{i0}}-a_{i},\text{ \ }\underline{%
\bar{C}_{00}}^{\prime }:=\underline{\bar{C}_{00}}-\alpha ,
\end{equation*}
with $a_{i}$ and $\alpha $\ arbitrary, induce on $C_{O}$ the same quadratic
form $\tilde{C}$, i.e. $\underline{\tilde{C}_{ij}}^{\prime }\equiv
\underline{\tilde{C}_{ij}}.$
\end{lemma}

\begin{proof}
Elementary calculation using the identity written above.
\end{proof}

\bigskip

In what follows, to simplify computations we will make the
restrictive condition that
\begin{equation}
\underline{C_{0i}}=0,\text{ \ }\underline{C_{00}}=-1,\text{ \ i.e \ \ }%
C:=-(dy^{0})^{2}+\underline{C_{ij}}dy^{i}dy^{j}\;.  \label{4.3}
\end{equation}
The set $\{y^{0}=r\}$  {is then a null cone for the metric}
$C,$ {with generator }$\underline{\ell ^{0}}=1,$
$\underline{\ell ^{i}}=\frac{y^{i}}{r}, $  {if and only if}
\begin{equation}
 {y^{i}\underline{\bar{C}_{ij}}=y^{j}}  \label{4.4}
\end{equation}%
(compare~\cite{ChJezierskiCIVP}).

The general relation between components in coordinates $y$ and
adapted null coordinates $x^{1}$, $x^{A}$ gives
\begin{equation*}
\tilde{C}\equiv \tilde{C}_{AB}dx^{A}dx^{B}\text{ \ with \ }\tilde{C}%
_{AB}\equiv \bar{C}_{AB},\text{ \ \ }\tilde{C}_{1A}=\tilde{C}_{11}=0,
\end{equation*}%
with $\tilde{C}_{AB}$ the components of a $x^{1}\equiv r$--dependent
Riemannian metric on the sphere $S^{n-1}$
\begin{equation*}
-(\underline{C_{00}}+\frac{y^{i}}{r}\underline{C_{0i}})\equiv C_{01}=1,\text{
\ \ }C_{0A}\equiv -\frac{\partial y^{i}}{\partial x^{A}}\underline{C_{0i}}=0%
\text{ \ \ }\underline{C_{00}}\equiv C_{00}=-1.\;
\end{equation*}%
This metric $C$ is also such that
\begin{equation*}
C^{00}\equiv C^{0A}\equiv C^{1A}\equiv 0,\quad C^{01}\equiv C^{11}\equiv 1,
\end{equation*}%
while $\bar{C}^{AB}$\ are the elements of the inverse of the
positive definite quadratic form with components
$\bar{C}_{AB}$.

\subsection{Full characteristic data}

We have seen in I that the trace $\bar{g}$ of a Lorentzian metric $g$
satisfying the reduced Einstein equations is a solution of the full Einstein
equations if and only if it satisfies the wave \ map gauge constraints.
These constraints $\mathcal{C}_{\alpha }=0$ are deduced in vacuum from the
identity satisfied by the Einstein tensor $S:$%
\begin{equation*}
\bar{\ell}^{\beta }\bar{S}_{\alpha \beta }\equiv \mathcal{C}_{\alpha }+%
\mathcal{L}_{\alpha }
\end{equation*}
where $\mathcal{L}_{\alpha }$ is linear and homogeneous in the wave gauge
vector $H$ while $\mathcal{C}_{\alpha }$ depends only on $\bar{g}$ and its
derivatives among $C_{O}$ and the given target $\hat{g}$. Given $\tilde{g}$,
i.e. $\bar{g}_{AB}\equiv \bar{C}_{AB}$, $\bar{g}_{1A}=\bar{g}_{11}=0$, the
remaining components $\nu_{0}\equiv \bar{g}_{01}$, $\nu_{A}\equiv \bar{g}%
_{0A}$, $\bar{g}_{00}$ are determined by the constraints and limit
conditions at the vertex $O$ which can always be satisfied by choice of
coordinates (see I). The Cagnac-Dossa theorem applies to components in the $%
y $ coordinates. Lemma \ref{lemma3.1} gives
%
\begin{equation}
\underline{\bar{g}_{00}}\equiv \bar{g}_{00},\quad \underline{\bar{g}_{0i}}%
\equiv -(\bar{g}_{00}+\nu _{0})r^{-1}y^{i}-\underline{\nu _{i}},\quad \text{%
with}\quad \underline{\nu _{i}}:=\nu _{A}\frac{\partial x^{A}}{\partial y^{i}%
}\;,  \label{4.5}
\end{equation}
\begin{equation*}
\underline{\bar{g}_{ij}}=(\bar{g}_{00}+2\nu_{0})r^{-2}y^{i}y^{j}
+r^{-1}(y^{i}\underline{\nu_{j}}+y^{j}\underline{\nu_{i}}) +\bar{g}_{AB}%
\frac{\partial x^{A}}{\partial y^{i}}\frac{\partial x^{B}}{\partial y^{j}}
\;,
\end{equation*}
while, for the chosen metric $C$ and $\bar{g}_{AB}\equiv \bar{C}_{AB}$
\begin{equation*}
\underline{\bar{C}}_{ij}\equiv r^{-2}y^{i}y^{j}+\bar{g}_{AB}\frac{\partial
x^{A}}{\partial y^{i}}\frac{\partial x^{B}}{\partial y^{j}} \;.
\end{equation*}
Therefore
\begin{equation}
\underline{\bar{g}_{ij}}=\underline{\bar{C}_{ij}}+(\bar{g}_{00}+2\nu
_{0}-1)r^{-2}y^{i}y^{j}+r^{-1}(y^{i}\underline{\nu _{j}}+y^{j}\underline{\nu
_{j}}) \;.  \label{4.6}
\end{equation}

\section{Null second fundamental form}

We have defined in I the \emph{null second fundamental form} of $(C_{O},%
\tilde{g})$ as the tensor $\chi$ on $C_{O}$ defined by the Lie derivative%
\footnote{{\footnotesize Recall that in arbitrary coordinates $x^{I}$ the
Lie derivative reads
\begin{equation*}
(\mathcal{L}_{\ell}\tilde{C})_{HK}\equiv \ell^{I}\partial_{I}\tilde{C}_{HK} +%
\tilde{C}_{HI}\partial_{K}\ell^{I} +\tilde{C}_{KI}\partial_{H}\ell^{I}.
\end{equation*}
}} with respect to the vector $\ell$ of the degenerate quadratic form $%
\tilde{g}$, namely in the coordinates $x^{1},x^{A}:$
\begin{equation}
\chi_{AB}:=\frac{1}{2}(\mathcal{L}_{\ell}\tilde{g})_{AB}\equiv \frac{1}{2}%
\partial_{1}\bar{g}_{AB} ,  \label{5.1}
\end{equation}
\begin{equation}
\chi_{A1}:=\frac{1}{2}(\mathcal{L}_{\ell}\tilde{g})_{A1}=0, \qquad
\chi_{11}:=\frac{1}{2}(\mathcal{L}_{\ell}\tilde{g})_{11}=0.  \label{5.2}
\end{equation}

In view of the application of the Cagnac-Dossa theorem we look for smooth
extensions. We define a smooth spacetime vector field $L$, vanishing at $O$
and with trace colinear with $\ell =\frac{\partial }{\partial x^{1}}$ on $%
C_{O}$, by its components respectively in the $x^{\alpha}$ and $y^{\alpha}$
coordinates:
\begin{equation*}
L:=y^{\lambda}\frac{\partial}{\partial y^{\lambda}} \equiv x^{0}\frac{%
\partial }{\partial x^{0}} +x^{1}\frac{\partial }{\partial x^{1}}, \quad
\text{hence} \quad \bar{L} \equiv x^{1}\ell \equiv r\ell.
\end{equation*}

We assume that the metric $C$ is smooth in $U$, a neighbourhood of $O$ in $%
\mathbf{R}^{n+1}$, i.e. its components $\underline{C_{ij}}$ are of class $%
C^{m}$, with $m$ as large as necessary in the considered context, functions
of the $y^{\alpha }$. We define a symmetric $C^{m-1}$ 2-tensor $X$,
identically zero in the case where $C\equiv \eta $, the Minkowski metric,
by:
\begin{equation}
X:=\frac{1}{2}\mathcal{L}_{L}C-C.\;  \label{5.3}
\end{equation}
In $y$ coordinates one has\footnote{{\footnotesize Recall that we underline
components in the $y$ coordinates and overline restrictions to $C_{O}.$}},
using $\ y^{i}\underline{C_{ij}}=y^{j}$
\begin{equation}
\underline{X}_{00}\equiv \underline{X}_{0i}\equiv 0,\qquad \underline{X}%
_{ij}\equiv \frac{1}{2}\{y^{0}\underline{\partial _{0}C_{ij}}+y^{h}%
\underline{\partial _{h}C_{ij}}\}.  \label{5.4}
\end{equation}
with
\begin{equation*}
y^{i}\underline{\partial _{0}C_{ij}}=0,
\end{equation*}
and, using $\partial _{h}y^{i}=\delta _{h}^{i}$%
\begin{equation}
y^{i}y^{h}\underline{\partial _{h}C_{ij}}=y^{h}\underline{\partial _{h}}%
(y^{i}\underline{C_{ij}})-y^{h}\underline{C_{ij}}\underline{\partial _{h}}%
y^{i}=0,  \label{5.5}
\end{equation}
which imply
\begin{equation}
\underline{L^{i}X_{ij}}\equiv y^{i}\underline{X_{ij}}=0.  \label{5.6}
\end{equation}

In $x$ coordinates we find, using the values of the components $C_{0\alpha }$
and $C_{1\alpha }$ of the metric $C$, that the tensor $X$ obeys the key
properties
\begin{equation}
X_{\mu 0}=0,\qquad X_{\mu 1}=0,  \label{5.7}
\end{equation}
while
\begin{equation}
X_{AB}\equiv \frac{1}{2}(x^{0}\partial _{0}C_{AB}+x^{1}\partial
_{1}C_{AB})-C_{AB}.  \label{5.8}
\end{equation}
Hence $X_{AB}$ reduces on the null cone $C_{O}$ to
\begin{equation}
\bar{X}_{AB}\equiv \frac{1}{2}x^{1}\partial_{1}\bar{C}_{AB}-\bar{C}%
_{AB}\equiv x^{1}\chi_{AB}-\bar{g}_{AB}.  \label{5.9}
\end{equation}
\emph{We still denote by $X$ the mixed $C^{m-1}$ tensor on spacetime
obtained from $X$ by lifting an index with the contravariant associate of $%
C; $} its $y$ components are the $C^{m-1}$ functions
\begin{equation*}
2\underline{X_{\alpha }^{\gamma }}\equiv \underline{C}^{\gamma \beta }\{y^{0}%
\underline{\partial _{0}C_{\alpha \beta }}+y^{i}\underline{\partial
_{i}C_{\alpha \beta }}\},
\end{equation*}
hence, $C$ being given by (\ref{4.3}),
\begin{equation}
\underline{X_{i}^{j}}\equiv \frac{1}{2}\underline{C^{jh}}\{y^{0}\underline{%
\partial _{0}C_{ih}}+y^{k}\underline{\partial _{k}C_{ih}}\},\quad \underline{%
X_{i}^{0}}\equiv \underline{X_{0}^{j}}\equiv \underline{X_{0}^{0}}\equiv 0,
\label{5.10}
\end{equation}
and
\begin{equation}
\underline{X_{i}^{j}L_{j}}\equiv 0.  \label{5.11}
\end{equation}
where the index of $L$ has been lowered with the metric $C$, so that this is
equivalent to (5.6). In $x$ coordinates $X_{A}^{C}$ are the only non
vanishing components of $X$. Their traces on $C_{O}$ are
\begin{equation*}
\bar{X}_{A}^{C}\equiv \frac{1}{2}x^{1}\bar{g}^{BC}\partial_{1}\bar{g}%
_{AB}-\delta_{A}^{C}\equiv x^{1}\chi_{A}^{C}-\delta_{A}^{C},
\end{equation*}
hence
\begin{equation}
\chi_{A}^{C}:=\frac{1}{2}\bar{g}^{BC}\partial _{1}\overline{g_{AB}}=\frac{1}{%
x^{1}}(\bar{X}_{A}^{C}+\delta _{A}^{C}),  \label{5.14}
\end{equation}
and
\begin{equation}
\tau :=\frac{1}{2}\bar{g}^{AB}\partial _{1}\overline{g_{AB}}=\frac{\overline{%
\mathrm{tr}X}}{x^{1}}+\frac{n-1}{x^{1}}.  \label{5.15}
\end{equation}
The trace of the tensor $X$ is the $C^{m-1}$ function
\begin{equation}
\mathrm{tr}X\equiv \underline{X_{\alpha }^{\alpha }}\equiv X_{\lambda
}^{\lambda }\equiv C^{AB}X_{AB}.  \label{5.12}
\end{equation}
On the light cone $C_{O}$ it holds that
\begin{equation}
\overline{\mathrm{tr}X}\equiv \underline{\bar{X}_{i}^{i}}\equiv \bar{g}^{AB}%
\bar{X}_{AB}\equiv \frac{x^{1}}{2}\bar{g}^{AB}\partial _{1}\bar{g}%
_{AB}-(n-1),  \label{5.13}
\end{equation}
\begin{equation}
|\chi|^{2}:=\chi_{A}^{C}\chi_{C}^{A}\equiv \frac{1}{(x^{1})^{2}}\{\bar{X}%
_{\beta }^{\alpha }\bar{X}_{\alpha }^{\beta }+2\overline{\mathrm{tr}X}+n-1\}.
\end{equation}

\section{A criterium: admissible series}

To show that the integration of the constraints, which appear as ODE in $%
x^{1}$, leads to traces on the cone of smooth spacetime
functions we shall use the following lemma, introduced by
Cagnac {(unpublished)} for formal series, but used here for
real analytic functions, a special class $C^{\omega }$ of
$C^{\infty }$ functions.

\begin{lemma}
A function is the trace $\bar{f}$ on $C_{O}^{T}$ of a spacetime function $f$
analytic in $U\cap Y_{O}^{T}$, $U$ a neighbourhood of $O$, if and only if it
admits on $U\cap C_{O}^{T}$ a convergent expansion of the form
\begin{equation}
\bar{f}\equiv f_{0}+\sum_{p=1}^{\infty }\bar{f}_{p}r^{p}  \label{6.1}
\end{equation}
with
\begin{equation}
\bar{f}_{p}\equiv \bar{f}_{p,i_{1}...i_{p}}\Theta^{i_{1}}...\Theta^{i_{p}}+%
\bar{f}^{\prime }{}_{p,i_{1}...i_{p-1}}\Theta^{i_{1}}...\Theta^{i_{p-1}}
\label{6.2}
\end{equation}
where $f_{0}$, $f_{p,i_{1}...i_{p}}$ and
$f_{p,i_{1}...i_{p-1}}^{\prime}$ are numbers. Such a series is
called an \emph{admissible series}. A coefficient $f_{p}$ of
the form (\ref{6.2}) is called an \emph{admissible coefficient
of order $p$}.
\end{lemma}

\begin{proof}
If $f$ is analytic it admits an expansion in Taylor series
\begin{equation}
f\equiv \sum_{p=0}^{\infty}
f_{\alpha_{1}...\alpha_{p}}y^{\alpha_{1}}...y^{\alpha_{p}}, \quad
f_{\alpha_{1}...\alpha_{p}}:=\frac{1}{p!} \frac{\partial^{p}f}{\partial
y^{\alpha_{1}}...\partial y^{\alpha_{p}}}(O).  \label{6.3}
\end{equation}
One goes from the formulas (\ref{6.3}) to (\ref{6.1}, \ref{6.2}) by
replacing $y^{i}$ by $r\Theta^{i}$ and $y^{0}$ by $r$, and conversely, in $%
\Omega \cap C_{O}$ or in $\Omega $.
\end{proof}

\begin{remark}
The identity (\ref{6.1}) is equivalent to saying that $\bar{f}$ is of the
form $\bar{f}=f_{1}+rf_{2}$, with $f_{1}$ and $f_{2}$ analytic functions of $%
y^{i}$.
\end{remark}

We say that an admissible series is of \emph{minimal order} $q$ if the
coefficients $f_{p}$ are identically zero for $p<q$.

\begin{proposition}
If the metric $C$ is analytic and satisfies the conditions (\ref{4.3}), (\ref%
{4.4}) then the functions $\overline{\mathrm{tr}X}$ and $|X|^{2}$ are
admissible series of minimal orders respectively 2 and 4.
\end{proposition}

The following lemmas will be very useful when integrating the constraints.

\begin{lemma}
If $f_{p}$ and $h_{q}$ are admissible coefficients of order respectively $p$
and $q$, then $f_{p}+h_{p}$ and $f_{p}h_{q}$ are admissible coefficients of
order respectively $p$ and $p+q$.
\end{lemma}

\begin{proof}
Elementary computation of ($f_{p}+h_{p})r^{p}$ and $f_{p}h_{q}r^{p+q}$
replacing $r\Theta ^{i}$ by $y^{i}$ and $r^{2}$ by $\Sigma _{i}(y^{i})^{2}$.
\end{proof}

\medskip

Suppose that $f$ and $h$ are admissible series of minimal
orders $q_{f}$ and $q_{h}$. The following are easy-to-check
consequences of the lemma:

\begin{itemize}
\item 1) $fh$ is an admissible series of minimal order $q_{f}+q_{h}$;

\item 2) if $q_{f}=q_{h}$ then $f+h$ is an admissible series of the same
minimal order;

\item 3) if $f(0)\not=0$ and $q_{f}=0$ then $1/f$ is an admissible series
also minimal order 0;

\item 4) $r\partial _{1}f$ is an admissible series of minimal order $q_{f}$,
unless $q_{f}=0$ and then it has a larger minimal order.
\end{itemize}

\begin{lemma}
\label{lemma6.5} If $k$ and $h$ are admissible series with $h$ of minimal
order $q_{h}\geq 1$ and the constant $k_{0}\equiv k(0)\geq 0$ then the ODE
\begin{equation}
r\partial _{1}f+kf=h  \label{6.4}
\end{equation}
admits one and only one solution $f$ which is also an admissible series of
the same minimal order $q_{h}$ as $h$. The result extends to $q_{h}=0$ if $%
k_{0}>0$.
\end{lemma}

\begin{proof}
Expand
\begin{equation}
f=\sum_{p=0}^{\infty }f_{p}r^{p},\qquad k=\sum_{p=0}^{\infty
}k_{p}r^{p},\qquad h=\sum_{p=q_{h}}^{\infty }h_{p}r^{p},
\end{equation}
hence
\begin{equation*}
r\partial _{1}f=\sum_{p=1}^{\infty }pf_{p}r^{p},
\end{equation*}
plug into the ODE (\ref{6.4}) and proceed to identifications.

We obtain by equating to zero the constant term
\begin{equation}
k_{0}f_{0}=h_{0},
\end{equation}
a relation which can be satisfied when $h_{0}\not=0$ only when
$k_{0}\not=0$.
We first consider the case where $h_{0}=0$, i.e. $q_{h}\geq 1$, and take $%
f_{0}=0$. We get the successive equalities
\begin{equation}
f_{1}+k_{0}f_{1}=h_{1},\quad \text{i.e.}\quad f_{1}=\frac{h_{1}}{1+k_{0}},
\end{equation}
and the recurrence relation, using $f_{0}=0$,
\begin{equation}
(p+k_{0})f_{p}+\sum_{q=1}^{p-1}k_{q}f_{p-q}=h_{p}\;.
\end{equation}
For $p<q_{h}$ we have $h_{p}=0$ and the recurrence relation gives $f_{p}=0$.
Therefore the leading admissible coefficients of $f$ and $h$ are always
related by
\begin{equation}
f_{q_{h}}=\frac{h_{q_{h}}}{q_{h}+k_{0}}.  \label{leading}
\end{equation}

We assume the series for $k$ and $h$ converge for all directions $\Theta^{i}$
and radius $cr<1$; that is, we assume that there exists a constant $c$ such
that
\begin{equation}
|k_{p}|<c^{p},\qquad \text{and}\qquad |h_{p}|<c^{p},
\end{equation}
Since $k_{0}\geq 0$ we have
\begin{equation*}
|f_{1}|\leq |\frac{h_{1}}{1+k_{0}}| < \frac{c}{1+k_{0}}\leq c .
\end{equation*}
Assume now that
\begin{equation}
|f_{p}|<c^{p}\qquad \text{for}\qquad p<p_{0},
\end{equation}
then from the iteration we get, for larger values of $p$, the inequality
\begin{equation}
|f_{p}|<\frac{p}{p+k_{0}}c^{p}\leq c^{p}\;.
\end{equation}
The bounds on $|f_{p}|$ show that the series for $f$ also converges. It is
an admissible series of minimal order $q_{f}=q_{h}$.

When $q_{h}=0$, i.e. $h_{0}\not=0$ and $k_{0}\not=0$ we take
\begin{equation*}
f_{0}=\frac{h_{0}}{k_{0}}
\end{equation*}
and we set
\begin{equation*}
F=f-f_{0}.
\end{equation*}
It satisfies the equation
\begin{equation}
r\partial_{1}F+kF=H, \qquad \text{with} \qquad H:=h-kf_{0} .
\end{equation}
We have
\begin{equation*}
H_{0}=0
\end{equation*}
and we apply to $F$ the previous result.
\end{proof}

\begin{corollary}
\label{corollary6.6} If $f$ and $h$ are admissible series related by (\ref%
{6.4}) and $p+k_{0}>0$, and $r^{-p}h$ is an admissible series then $r^{-p}f$
is an admissible series of the same minimal order.
\end{corollary}

\begin{proof}
Set $f=r^{p}\phi$. If $f$ satisfies (\ref{6.4}) then $\phi$ satisfies the
equation
\begin{equation*}
r\partial_{1}\phi+(p+k)\phi=r^{-p}h.
\end{equation*}
\end{proof}

\begin{remark}
The following example is a case of a differential equation of the form (\ref%
{6.4}) with $q_{h}=1$, but $k_{0}$ a negative integer, which does not admit
as a solution an admissible series. Let
\begin{equation*}
k=\frac{1}{r-1}=-1-r-r^{2}-r^{3}-...,\qquad h=\frac{r}{1-r^{2}}%
=r+r^{3}+r^{5}+...,
\end{equation*}
We can solve the ODE explicitly,
\begin{equation*}
f=\frac{r}{r-1}(f_{\infty }+\log \frac{r+1}{r})
\end{equation*}
with $f_{\infty }$ an arbitrary integration constant, which cannot be
expanded in powers of $r$ near $0$. However if we change $k$ to $r/(r-1)$,
then $k_{0}$ changes from $-1$ to $0$, the problem disappears. Remark that
the problem also disappears if we change $h$ to $r^{2}/(1-r^{2})$, i.e. $%
q_{h}=2$. 
\end{remark}

In the following we will assume the metric $C$, of the form (\ref{4.3}) and
satisfying (\ref{4.4}) is analytic, takes Minkowskian values at the vertex $%
O $, and is such the components of its trace on $C_{O}$ satisfy
\begin{equation}
\underline{\bar{C}_{ij}}\equiv \delta_{ij}+\underline{\bar{c}_{ij}}, \qquad
\underline{\bar{C}^{ij}}\equiv \delta^{ij}+\underline{\bar{c}^{ij}}\; ,
\end{equation}
where $\underline{\bar{c}_{ij}}$ and $\underline{\bar{c}^{ij}}$ have
admissible expansions of minimal order 2 while $\overline{\underline{%
\partial_{0}c_{ih}}}$ has an admissible expansion of minimal order 1. The
definition (\ref{5.10}) implies then that
\begin{equation}
\underline{\bar{X}_{i}^{j}}\equiv \frac{1}{2}\underline{C^{jh}}\{r\overline{%
\underline{\partial _{0}c_{ih}}}+y^{k}\overline{\underline{\partial
_{k}c_{ih}}}\}
\end{equation}
has an admissible expansion of minimal order 2.

\section{The first wave-map gauge constraint}

We have deduced our first constraint\footnote{%
See I.} from the identity
\begin{equation}
\bar{\ell}^{\beta }\bar{S}_{1\beta }\equiv \bar{R}_{11}\equiv -\partial
_{1}\tau +\nu^{0}\partial_{1}\nu_{0}\tau -\frac{1}{2}\tau (\bar{\Gamma}_{1}
+\tau)-\chi_{A}^{B}\chi_{B}^{A},  \label{7.1}
\end{equation}
with
\begin{equation}
\bar{\Gamma}_{1}\equiv \bar{W}_{1}+\bar{H}_{1},\qquad \bar{W}_{1}\equiv -\nu
_{0}\bar{g}^{AB}rs_{AB}.  \label{7.2}
\end{equation}
Hence for the first wave-map gauge constraint in vacuum we have the equation
\begin{equation}
\mathcal{C}_{1}:=-\partial _{1}\tau +\nu ^{0}\partial _{1}\nu _{0}\tau -%
\frac{1}{2}\tau (\tau -\nu _{0}\bar{g}^{AB}rs_{AB})-\chi_{A}^{B}
\chi_{B}^{A}=0.  \label{7.3}
\end{equation}
When $\bar{g}_{AB}$ is known this equation reads as a first order
differential equation for $\nu _{0}$%
\begin{equation}
\nu^{0}\partial_{1}\nu_{0}=\tau^{-1}\partial_{1}\tau +\frac{1}{2}(\tau
-\nu_{0}\bar{g}^{AB}rs_{AB})+\tau^{-1}\chi_{A}^{B}\chi_{B}^{A}.  \label{7.4}
\end{equation}
It can be written as a linear equation for $\nu ^{0}-1$,
\begin{equation}
\partial _{1}\nu ^{0}+a(\nu ^{0}-1)+b=0,  \label{7.5}
\end{equation}
with
\begin{equation}
a:=\tau ^{-1}\partial _{1}\tau +\frac{1}{2}\tau +\tau ^{-1}|\chi|^{2},
\qquad |\chi|^{2}\equiv \chi_{A}^{B}\chi_{B}^{A},  \label{7.6}
\end{equation}
\begin{equation}
b:=a-\frac{1}{2}\bar{g}^{AB}rs_{AB}.  \label{7.7}
\end{equation}
In the flat case $\bar{g}^{AB}=\eta^{AB}$, $\tau =\frac{n-1}{r}$, $%
\chi_{A}^{B}=\frac{1}{r}\delta_{A}^{B}$ the equation reduces to:
\begin{equation*}
\partial_{1}\nu^{0}+\frac{1}{2}(\nu^{0}-1)\frac{n-1}{r}=0;
\end{equation*}
it has one solution tending to $1$ when $r$ tends to zero, $\nu_{0}=1$. In
the general case (\ref{7.5}) reads, with $f:=\nu^{0}-1$,
\begin{equation}
r\partial_{1}f+kf+h=0, \qquad k:=ar, \quad h:=br=ar-\frac{1}{2}\bar{g}%
^{AB}r^{2}s_{AB} .  \label{7.8}
\end{equation}
Recall that $x^{1}\equiv r$ and
\begin{equation}
\chi_{A}^{C}\equiv \frac{1}{2}\bar{g}^{BC}\partial_{1}\overline{g_{AB}}\;=
\frac{1}{r}(\bar{X}_{A}^{C}+\delta_{A}^{C}).  \label{7.9}
\end{equation}
Hence
\begin{equation*}
|\chi|^{2}=\frac{|\bar{X}|^{2}+2\overline{\mathrm{tr}X}+n-1}{r^{2}},
\end{equation*}
\begin{equation}
\tau =\frac{\overline{\mathrm{tr}X}}{r}+\frac{n-1}{r}, \qquad \tau^{-1}=%
\frac{r}{n-1+\overline{\mathrm{tr}X}},  \label{7.10}
\end{equation}
where $\overline{\mathrm{tr}X}$ is an admissible series of minimal order 2.
The function $\{1+\frac{1}{n-1}\overline{\mathrm{tr}X}\}^{-1}$ is the trace
of a $C^{\omega }$ function as long as $1+\frac{1}{n-1}\mathrm{tr}X$ does
not vanish, hence always in a neighbourhood of $O$ since $\mathrm{tr}X$
vanishes there.

It holds that
\begin{equation}
\partial_{1}\tau \equiv -\frac{n-1+\overline{\mathrm{tr}X}}{r^{2}}+ \frac{%
\partial_{1}\overline{\mathrm{tr}X}}{r},  \label{7.11}
\end{equation}
\begin{equation}
\tau^{-1}\partial_{1}\tau \equiv -\frac{1}{r}+ \frac{\partial_{1}\overline{%
\mathrm{tr}X}}{n-1+\overline{\mathrm{tr}X}}.  \label{7.12}
\end{equation}
Also we can write
\begin{equation}
\tau^{-1}|\chi|^{2}\equiv \frac{|\bar{X}|^{2}+2\overline{\mathrm{tr}X}+n-1} {%
r(n-1+\overline{\mathrm{tr}X})} \equiv \frac{|\bar{X}|^{2}+\overline{\mathrm{%
tr}X}} {r(n-1+\overline{\mathrm{tr}X})}+\frac{1}{r}.  \label{7.13}
\end{equation}
Finally computation gives
\begin{equation}
k\equiv ar\equiv \frac{n-1}{2}+\frac{r\partial_{1}\overline{\mathrm{tr}X}+
\overline{\mathrm{tr}X}+|\bar{X}|^{2}}{n-1+\overline{\mathrm{tr}X}}.
\label{7.14}
\end{equation}

We see that $k-\frac{n-1}{2}$ admits in a neighbourhood of $O$ an admissible
development of minimal order 2.

On the other hand, since in the $x$ coordinates $\eta_{1\alpha}=0$, $%
r^{2}s_{AB}=\eta_{AB}$ and we have assumed $\bar{C}^{00}=\bar{C}^{0A}=0$, $%
\bar{C}^{01}=1$, $\bar{g}^{AB}=\bar{C}^{AB}$, we have
\begin{equation*}
\bar{g}^{AB}r^{2}s_{AB}\equiv \bar{g}^{AB}\eta _{AB}\equiv \bar{C}^{AB}\eta
_{AB}\equiv \bar{C}^{\alpha \beta }\eta _{\alpha \beta }-2.
\end{equation*}
Hence, using now the values of $\underline{C_{\alpha\beta}}$
\begin{equation}
\frac{1}{2}\bar{g}^{AB}r^{2}s_{AB}\equiv \frac{1}{2}\overline{\underline{%
(1+C^{ij}\delta _{ij}-2)}}\equiv \frac{n-1}{2}+\frac{1}{2}\overline{%
\underline{c^{ij}\delta _{ij}}},  \label{7.15}
\end{equation}
where $\overline{\underline{c^{ij}\delta_{ij}}}$ has an admissible
development of minimal order $2$. We conclude that
\begin{equation}
h\equiv \frac{r\partial _{1}\overline{\mathrm{tr}X}+\overline{\mathrm{tr}X}+|%
\bar{X}|^{2}}{(n-1+\overline{\mathrm{tr}X})}-\frac{1}{2}\overline{\underline{%
c^{ij}\delta _{ij}}}  \label{7.16}
\end{equation}
admits also such a development. Lemma \ref{lemma6.5} applies, and we have
proved:

\begin{theorem}
If the basic characteristic data are induced on $C_{O}$ by a $C^{\omega }$
(i.e. analytic) metric of the form (\ref{4.3}), hence satisfying (\ref{3.4}%
), then $\nu^{0}-1$ admits an admissible expansion of minimal order $2$,
hence is the trace in a neighbourhood of $O$ of a $C^{\omega}$ spacetime
function. Then $\nu^{0}=\bar{N}^{0}$ with $N^{0}\in C^{\omega}$ and $%
N^{0}(O)=1$. In a neighbourhood of $O$, it holds that $\nu_{0}=\bar{N}_{0}$,
$N_{0}=(N^{0})^{-1}$.
\end{theorem}

In the expression of the characteristic initial data in $y$ coordinates
appears $r^{-2}(\nu _{0}-1)$ which though continuous on each ray as $r$
tends to zero, is not an admissible expansion. We introduce the following
definition.

\begin{definition}
\label{Def} A metric $C$ satisfying (\ref{4.3})-(\ref{4.4}) is said to be
near-round at the vertex if there is a neighbourhood of $O$ where $r^{-1}%
\underline{\bar{c}}_{ij}\equiv \underline{\bar{d}}_{ij}$ and $r^{-2}%
\underline{c}_{ij}\delta ^{ij}\equiv \bar{D}$ \ with $\underline{\bar{d}}%
_{ij}$ and $\bar{D}$ admissible series.
\end{definition}

If $C$ is near-round at the vertex $\ \underline{\overset{\_}{C}}_{ij}$ has
an analytic extension $\underline{C_{ij}}$ of the form, with $\underline{%
d_{ij}}$ analytic extension of $\underline{\bar{d}_{ij}}$,
\begin{equation*}
\underline{C}_{ij}\equiv \delta _{ij}+\underline{c}_{ij}\text{ \ \ \ }%
\underline{c}_{ij}\equiv y^{0}\underline{d}_{ij},\text{ }
\end{equation*}%
therefore
\begin{equation*}
\underline{\partial _{h}C_{ij}}\equiv 0\text{ \ \ for \ \ }y^{0}=0,\text{ \
hence}\quad \underline{\partial _{h}C^{ij}}\equiv 0\text{ \ \ \ for \ \ }%
y^{0}=0,
\end{equation*}%
and
\begin{equation*}
\underline{C^{ij}}\equiv \delta ^{ij}+y^{0}\underline{d^{ij}},
\end{equation*}%
with $\underline{d^{ij}}$ some analytic functions. Hence
\begin{equation*}
\underline{C^{ij}c_{ij}}\equiv (y^{0})^{2}\{D+\underline{d^{ij}}\underline{%
d_{ij}}\}.
\end{equation*}
Using the definition of $\underline{X_{\alpha\beta}}$ we see
that if $C$ is near-round at the vertex, then
\begin{equation}
\underline{X_{\alpha\beta}}\equiv y^{0}\underline{Y_{\alpha\beta}}, \quad
\text{with} \quad \underline{Y}_{ij}:=\frac{1}{2}\{\underline{d_{ij}}+ y^{0}%
\underline{\partial_{0}d_{ij}}+ y^{h}\underline{\partial_{h}d_{ij}}\}, \quad
\underline{Y_{i0}}\equiv \underline{Y_{00}}\equiv 0.  \label{7.17}
\end{equation}
Hence
\begin{equation*}
\mathrm{tr}X\equiv y^{0}\,\mathrm{tr}Y, \quad \text{with} \quad \mathrm{tr}
Y\equiv \underline{C^{ij}}\underline{Y_{ij}}\equiv \frac{1}{2}\underline{%
C^{ij}}\{ \underline{d_{ij}}+ y^{0}\underline{\partial_{0}d_{ij}}+ y^{h}%
\underline{\partial_{h}d_{ij}}\}.
\end{equation*}
An elementary computation shows that $\mathrm{tr}Y$ is of the
following form, with $Z$ an analytic function,
\begin{equation*}
\mathrm{tr}Y\equiv y^{0}Z, \qquad \text{hence} \qquad \mathrm{tr}X\equiv
(y^{0})^{2}Z.
\end{equation*}
On the other hand
\begin{equation*}
\underline{X_{j}^{i}}\equiv \underline{C^{ih}}\underline{X_{jh}}= y^{0}
\underline{C^{ih}}\underline{Y_{jh}}:=y^{0}\underline{Y_{j}^{i}},
\end{equation*}
therefore
\begin{equation*}
|X|^{2}\equiv \underline{X_{j}^{i}}\underline{X_{i}^{j}}\equiv (y^{0})^{2}%
\underline{Y_{j}^{i}}\underline{Y_{i}^{j}}.
\end{equation*}

We deduce from these formulas that if $C$ is near-round at the vertex then $%
r^{-2}\overline{\mathrm{tr}X}\equiv \bar{Z}$ and $r^{-2}|\bar{X}|^{2}\equiv |%
\bar{Y}|^{2}$ are admissible series.

\begin{theorem}
A sufficient condition for $r^{-2}(\nu ^{0}-1)$, with $\nu _{0}$ solution of
the first wave-map gauge constraint, to have an admissible expansion is that
the $C^{\omega }$ metric $C$ given by (\ref{4.3}) which induces the basic
characteristic data be near-round at the vertex.
\end{theorem}

\begin{proof}
Since $f:=\nu_{0}-1$ satisfies the equation (\ref{7.8}), $%
\phi:=r^{-2}(\nu_{0}-1)$ satisfies
\begin{equation}
r\partial_{1}\phi +(2+k)\phi +r^{-2}h=0.
\end{equation}
The expression (\ref{7.16}) shows that for a metric $C$ round at the vertex $%
r^{-2}h$ admits an admissible development, the application of Corollary \ref%
{corollary6.6} gives the result.
\end{proof}

\section{The $C_{A}$ constraint}

We have written in I the $C_{A}$ constraint in vacuum as
\begin{equation*}
\mathcal{C}_{A}= -\frac{1}{2}(\partial_{1}\xi_{A}+\tau \xi_{A}) +\tilde{%
\nabla}_{B}\chi_{A}^{B} -\frac{1}{2}\partial_{A}\tau +\partial_{A}(\frac{1}{2%
}\bar{W}_{1}+\nu_{0}\partial_{1}\nu^{0}) ,
\end{equation*}
where $\xi_{A}$ is defined as
\begin{equation}
\xi_{A}:= -2\nu^{0}\partial_{1}\nu_{A} +4\nu^{0}\nu_{C}\chi_{A}^{C} +\left(
\bar{W}^{0}-\frac{2}{r}\nu^{0}\right) \nu_{A} +\bar{g}_{AB} \bar{g}%
^{CD}(S_{CD}^{B}-\tilde{\Gamma}_{CD}^{B})\;.  \label{8.1}
\end{equation}

Using the first constraint we find
\begin{equation}
\nu_{0}\partial_{1}\nu^{0}+\frac{1}{2}\bar{W}_{1}=-a.  \label{8.2}
\end{equation}
where $a$ is given by (\ref{7.14}), hence
\begin{equation*}
\partial_{A}(a+\frac{1}{2}\tau)\equiv r^{-1}\partial_{A}F(\overline{\mathrm{%
tr}X},|\bar{X}|^{2}),
\end{equation*}
where
\begin{eqnarray*}
F(\overline{\mathrm{tr}X},|\bar{X}|^{2}) &:=& \frac{r\partial_{1}\overline{%
\mathrm{tr}X} +\overline{\mathrm{tr}X} +|\bar{X}|^{2}} {n-1+\overline{%
\mathrm{tr}X}} +\frac{1}{2}\overline{\mathrm{tr}X} \\
&\equiv& \frac{r\partial_{1}\overline{\mathrm{tr}X} +\frac{1}{2}\{(n+1)%
\overline{\mathrm{tr}X} +|\overline{\mathrm{tr}X}|^{2}\} +|\bar{X}|^{2}} {%
n-1+\overline{\mathrm{tr}X}}
\end{eqnarray*}
admits in a neighbourhood of $O$ an admissible development of minimal order
2.

We have:
\begin{equation}
\mathcal{C}_{A}\equiv -\frac{1}{2}(\partial_{1}\xi_{A}+\tau \xi_{A})+ \tilde{%
\nabla}_{B}\chi_{A}^{B}-r^{-1}\partial_{A}F(\overline{\mathrm{tr}X} ,|\bar{X}%
|^{2})\;=0.  \label{8.3}
\end{equation}

\subsection{Equations for $\protect\xi_{A}$}

We set $\xi _{1}=\xi _{0}=0$ on the cone and we define $\underline{\xi }_{i}$
by
\begin{equation}
\underline{\xi }_{i}:=\frac{\partial x^{\alpha }}{\partial y^{i}}\xi
_{\alpha }\equiv \frac{\partial x^{A}}{\partial y^{i}}\xi _{A}.  \label{8.4}
\end{equation}
It holds that
\begin{equation}
y^{i}\underline{\xi }_{i}=0,  \label{8.5}
\end{equation}
because (recall that $x^{1}=r$, $y^{i}\equiv r\Theta ^{i}(x^{A})$)
\begin{equation}
y^{i}\frac{\partial x^{A}}{\partial y^{i}}\equiv r\frac{\partial y^{i}}{%
\partial x^{1}}\frac{\partial x^{A}}{\partial y^{i}}\equiv r\delta
_{1}^{A}=0,\qquad \frac{\partial }{\partial r}\frac{\partial y^{i}}{\partial
x^{A}}\equiv \frac{1}{r}\frac{\partial y^{i}}{\partial x^{A}}.  \label{8.6}
\end{equation}
We have
\begin{equation}
\frac{\partial y^{i}}{\partial x^{A}}\underline{\xi }_{i}\equiv \frac{%
\partial y^{i}}{\partial x^{A}}\frac{\partial x^{B}}{\partial y^{i}}\xi
_{B}\equiv \delta_{A}^{B}\xi_{B}\equiv \xi_{A} .  \label{8.7}
\end{equation}
The equation (\ref{8.7}) implies that
\begin{equation*}
\partial _{1}\xi _{A}\equiv (\frac{\partial }{\partial r}\underline{\xi _{i}}%
+r^{-1}\underline{\xi _{i}})\frac{\partial y^{i}}{\partial x^{A}} ,
\end{equation*}
hence
\begin{equation*}
\mathcal{C}_{A}\equiv -\frac{1}{2}\frac{\partial y^{i}}{\partial x^{A}}\{%
\frac{\partial }{\partial r}\underline{\xi _{i}}+\underline{\xi _{i}}%
(r^{-1}+\tau )\}+\tilde{\nabla}_{B}\chi_{A}^{B}-r^{-1}\partial _{A}F(%
\overline{\mathrm{tr}X},|\bar{X}|^{2})=0.
\end{equation*}
Since $\mathrm{tr}X\equiv \underline{\mathrm{tr}X}$ is a scalar function and
the equation of $C_{O}$ in the $x$ coordinates is $x^{0}=0$ and $y^{0}$ does
not depend on $x^{A}$, it holds that
\begin{equation*}
\frac{\partial}{\partial x^{A}}\overline{\mathrm{tr}X}\equiv \overline{\frac{%
\partial}{\partial x^{A}}\mathrm{tr}X}\equiv \overline{\frac{\partial y^{i}}{%
\partial x^{A}} \frac{\partial}{\partial y^{i}}\underline{\mathrm{tr}X}}%
\equiv \frac{\partial y^{i}}{\partial x^{A}} \overline{\frac{\partial }{%
\partial y^{i}}\underline{\mathrm{tr}X}},
\end{equation*}
analogously
\begin{equation*}
\frac{\partial }{\partial x^{A}}\overline{|X|^{2}}\equiv \frac{\partial y^{i}%
}{\partial x^{A}}\overline{\frac{\partial }{\partial y^{i}}\underline{|X|^{2}%
}},\qquad \frac{\partial }{\partial x^{A}}\overline{F(\mathrm{tr}X,|X|^{2})}%
\equiv \frac{\partial y^{i}}{\partial x^{A}}\overline{\frac{\partial }{%
\partial y^{i}}\underline{F(\mathrm{tr}X,|X|^{2})}}.
\end{equation*}

We now compute, with covariant derivatives $\tilde{\nabla}$ taken in the
Riemannian metric $\bar{g}_{AB}\equiv \bar{C}_{AB}$
\begin{equation*}
\tilde{\nabla}_{B}\chi_{A}^{B}\equiv \tilde{\nabla}_{B}(\frac{1}{x^{1}}\bar{X%
}_{A}^{B} +\frac{1}{x^{1}}\delta_{A}^{B})\equiv \frac{1}{x^{1}}\tilde{\nabla}%
_{B}\bar{X}_{A}^{B}.
\end{equation*}
The Christoffel symbols $\tilde{C}_{BC}^{A}$ of the Riemannian connection $%
\tilde{\nabla}$ are equal (recall that $C^{B0}=C^{B1}=C^{00}=0)$ to the
trace on $C_{O}$ of the Christoffel symbols with the same indices of the
spacetime metric $C$, hence, denoting by $^{(C)}\nabla$ the covariant
derivative in the metric $C$
\begin{equation*}
\tilde{\nabla}_{B}\chi_{A}^{B}\equiv \frac{1}{x^{1}}\overline{%
^{(C)}\nabla_{B}X_{A}^{B}}.
\end{equation*}
Since the $X_{A}^{B}$ are the only non vanishing components of the tensor $%
X, $ and due to the form chosen for the metric $C$ we find that
\begin{equation*}
\overline{^{(C)}\nabla_{B}X_{A}^{B}}= \overline{^{(C)}\nabla_{\alpha}X_{A}^{%
\alpha}}= \frac{\partial y^{i}}{\partial x^{A}} \underline{\overline{%
^{(C)}\nabla_{\alpha }X_{i}^{\alpha}}}
\end{equation*}
and the equations (\ref{8.3}) read
\begin{equation}
\mathcal{C}_{A}\equiv \frac{\partial y^{j}}{\partial x^{A}}\frac{1}{r} %
\Big\{ -\frac{1}{2}[r\frac{\partial}{\partial r}\underline{\xi_{j}} +%
\underline{\xi_{j}}(n+\underline{\overline{\mathrm{tr}X}})] +\underline{%
\overline{ ^{(C)}\nabla_{\alpha}X_{j}^{\alpha}}} -\underline{\overline{\frac{%
\partial}{\partial y^{j}} F(\mathrm{tr}X,|X|^{2}}}) \Big\}=0.  \label{8.8}
\end{equation}
The parentheses constitute a linear diagonal operator on the $\underline{%
\xi_{j}}$ of the type considered in lemma \ref{lemma6.5}. Equating it to
zero gives an equation with \emph{solution} $\underline{\xi_{j}}$ \emph{an
admissible series of minimal order 1}. We denote by \underline{$\Xi_{j}$}
the extension of $\underline{\xi_{i}}$ to spacetime, that is we have
\begin{equation}
\underline{\xi_{i}}\equiv \overline{\underline{\Xi_{i}}} ,  \label{8.9}
\end{equation}
where $\underline{\Xi_{i}}$ are analytical functions beginning by linear
terms.

\subsection{Equations for $\protect\underline{\protect\nu_{i}}$}

We now consider the equations (\ref{8.1}) which read
\begin{equation}
\partial_{1}\nu_{A}+\left( \frac{1}{r}-\frac{1}{2}\bar{W}_{1}\right) \nu
_{A}-2\nu_{C}\chi_{A}^{C}-\frac{1}{2}\nu_{0}\bar{g}_{AB} \bar{g}%
^{CD}(S_{CD}^{B}-\tilde{\Gamma}_{CD}^{B})+\frac{1}{2}\nu_{0}\xi_{A}=0\;.
\label{8.10}
\end{equation}
We set
\begin{equation}
\underline{\bar{g}_{0i}}\equiv -\underline{\nu_{i}}+\lambda \underline{\bar{L%
}_{i}},\text{ \ \ \ with \ \ }\underline{\bar{L}_{i}}\equiv \underline{\bar{C%
}_{ij}}\underline{\bar{L}^{j}},\text{ \ }\underline{\bar{L}^{j}}\equiv y^{j},
\label{8.11}
\end{equation}
with $\underline{\nu_{i}}$ such that
\begin{equation}
\underline{\nu_{i}\bar{L}^{i}}\equiv \underline{\nu_{i}} y^{i}=0;
\label{8.12}
\end{equation}
that is, using (\ref{8.11}),
\begin{equation*}
\lambda \equiv (\underline{\bar{L}_{i}\bar{L}^{i}})^{-1}\underline{\bar{g}%
_{0j}\bar{L}^{j}}\equiv r^{-2}y^{j}\underline{\bar{g}_{0j}}.
\end{equation*}
Then  (compare (\ref{4.5}))
\begin{equation*}
\nu_{A}\equiv -\frac{\partial y^{i}}{\partial x^{A}}\underline{\bar{g}_{0i}}%
\equiv \frac{\partial y^{i}}{\partial x^{A}}\underline{\nu_{i}}.
\end{equation*}
Hence
\begin{equation*}
\partial_{1}\nu_{A}\equiv \frac{\partial y^{i}}{\partial x^{A}}(\partial _{1}%
\underline{\nu_{i}}+r^{-1}\underline{\nu_{i}}).
\end{equation*}
We recall that
\begin{equation*}
\chi_{A}^{C}\equiv \frac{1}{r}(\bar{X}_{A}^{C}+\delta_{A}^{C}) \quad \text{%
i.e.} \quad r\nu_{C}\chi_{A}^{C}\equiv \nu_{C}\bar{X}_{A}^{C}+\nu_{A}.
\end{equation*}
Therefore, after product by $r$, the equations can be written as follows
\begin{equation*}
\frac{\partial y^{i}}{\partial x^{A}}(r\partial_{1}\underline{\nu_{i}} -%
\frac{1}{2}r\bar{W}_{1}\underline{\nu_{i}} +\frac{1}{2}\nu_{0}r\underline{%
\xi_{i}}) -2\bar{F}_{A}-\frac{1}{2}r\nu_{0}\bar{E}_{A}=0 ,
\end{equation*}
with
\begin{equation*}
\bar{F}_{A}:=\nu_{C}\bar{X}_{A}^{C}, \qquad \bar{E}_{A}:=\bar{g}_{AB}\bar{g}%
^{CD}(S_{CD}^{B}-\tilde{\Gamma}_{CD}^{B}).
\end{equation*}
By definition it holds that
\begin{equation*}
\nu_{C}\bar{X}_{A}^{C}\equiv \overline{g_{0C}X_{A}^{C}},
\end{equation*}
and since $X_{A}^{C}$ are the only non vanishing components of the mixed
tensor $X$ in the coordinates $x$
\begin{equation*}
g_{0C}X_{A}^{C}=g_{0\lambda }X_{A}^{\lambda }\equiv \frac{\partial y^{\alpha
}}{\partial x^{0}}\frac{\partial y^{\beta }}{\partial x^{A}}\underline{%
g_{\alpha \lambda }}\underline{X_{\beta }^{\lambda }}
\end{equation*}
Recalling that $\underline{X_{i}^{0}}\equiv 0$ we find, (using (\ref{5.11}),
$\underline{\overline{X_{i}^{j}L_{j}}}=0)$
\begin{equation*}
\bar{F}_{A}\equiv \nu_{C}\bar{X}_{A}^{C}\equiv -\frac{\partial y^{i}}{%
\partial x^{A}}\, \underline{\overline{g_{0j}X_{i}^{j}}}\equiv \frac{%
\partial y^{i}}{\partial x^{A}} \, \underline{\overline{\nu_{j}X_{i}^{j}}} .
\end{equation*}

We then remark that $S_{CD}^{B}$ $-\tilde{C}_{CD}^{B}$ is the trace on $%
C_{O} $ of the difference of the components of the Christoffel symbols, $%
\eta_{\beta\gamma}^{\alpha}$ and $C_{\beta\gamma}^{\alpha}$, with these
angular $x$ indices of the Minkowski metric $\eta$ and the metric $C:$
\begin{equation*}
\bar{E}_{A}:=\bar{g}_{AB}\bar{g}^{CD}(S_{CD}^{B}-\tilde{C}_{CD}^{B}), \quad
\text{with} \quad E_{A}\equiv C_{AB}C^{CD}(\eta_{CD}^{B}-C_{CD}^{B}).
\end{equation*}
Using the expressions of $\eta$ and $C$ and the vanishing of the Christoffel
symbols of $\eta$ in the $y$ coordinates we find
\begin{equation*}
\bar{E}_{A} \equiv \overline{C_{A\alpha}C^{\lambda\mu}
(\eta_{\lambda\mu}^{\alpha}-C_{\lambda\mu}^{\alpha})} \equiv \frac{\partial
y^{i}}{\partial x^{A}} \overline{\underline{C_{i\alpha}C^{\lambda\mu}
(\eta_{\lambda\mu}^{\alpha}-C_{\lambda\mu}^{\alpha})}} \equiv -\frac{%
\partial y^{i}}{\partial x^{A}} \overline{\underline{C_{ij}C^{\lambda\mu}C_{%
\lambda\mu}^{j}}},
\end{equation*}
with $\underline{C_{\lambda\mu}^{j}}$ analytic functions, components of
Christoffel symbols of the metric $C$ in $y$ coordinates, that is, using the
values of the $y$ components of the metric $C$
\begin{equation*}
\underline{C_{ij}C^{\lambda \mu }C_{\lambda \mu }^{j}}\equiv \frac{1}{2}%
\underline{C^{hk}(\partial_{h}c_{ik}+\partial_{h}c_{ik}-\partial_{i}c_{hk})}.
\end{equation*}

We recall from (\ref{7.15}) that
\begin{equation*}
r\bar{W}_{1}\equiv -\nu_{0}\bar{g}^{AB}r^{2}s_{AB}\equiv -\nu_{0}\{n-1+
\overline{\underline{c^{ij}\delta_{ij}}}\},
\end{equation*}
and we find that the equations (\ref{8.10}) can be written
\begin{equation*}
\frac{\partial y^{i}}{\partial x^{A}}\mathcal{L}_{i}=0 \; ,
\end{equation*}
where $\mathcal{L}_{i}$ is the following linear operator on $\underline{%
\nu_{i}}$
\begin{equation}
r\partial_{1}\underline{\nu_{i}} +\nu_{0}\{\frac{n-1}{2} +\frac{1}{2}
\overline{\underline{c^{hk}\delta_{hk}}}\} \underline{\nu_{i}} -2\underline{%
\overline{\nu_{j}X_{i}^{j}}} +\frac{1}{2}\nu_{0}r\underline{\xi_{i}} -\frac{1%
}{2}r\nu_{0} \overline{\underline{C_{ij}C^{\lambda\mu}C_{\lambda\mu}^{j}}}
=0\;.  \label{8.13}
\end{equation}

We extend as follows Lemma~\ref{lemma6.5}.

\begin{lemma}
If $k_{i}^{j}$ and $h_{i}$ are admissible series of minimal orders $1$, and
the constant $k_{0}\geq 0$ then the ODE
\begin{equation}
r\partial_{1}\underline{\nu_{i}}+k_{0}\underline{\nu_{i}}+k_{i}^{j}
\underline{\nu_{j}}=h_{i}  \label{8.14}
\end{equation}
admits a solution $\underline{\nu_{i}}$ which is also an admissible series
of the same minimal order than $h$.
\end{lemma}

Recalling that $\nu_{0}-1$ is an admissible series of minimal order 2 we see
that this lemma applies to (\ref{8.13}). We have proved:

\begin{theorem}
If the basic characteristic data is induced by an analytic metric $C$
satisfying (\ref{4.3})-(\ref{4.4}), the equation (%
\ref{8.13}) admits one and only one solution $\underline{\nu_{i}}$ which is
an admissible series of minimal order 2. We denote by $\underline{N_{i}}$
its spacetime extension.
\end{theorem}

We now prove:

\begin{theorem}
A sufficient condition for $r^{-2}\underline{\nu _{i}}$ to have an
admissible expansion is that the $C^{\omega }$ metric $C$ given by (\ref{4.3}%
) which induces the basic characteristic data be near-round at the vertex.
\end{theorem}

\begin{proof}
Using the relation between $X$ and $Y$ the linearity of $F$ in $\mathrm{tr}X$
and $|X|^{2}$ we see that for $C$ round at the vertex the equation satisfied
by $\underline{\xi_{j}}$ reads
\begin{equation*}
\frac{1}{2}\{r\frac{\partial }{\partial r}\underline{\xi_{j}}+\underline{%
\xi_{j}}(n+r^{2}\bar{Z})\}=h_{j},
\end{equation*}
with
\begin{equation*}
h_{j}:=\underline{\overline{^{(C)}\nabla_{\alpha }((y^{0})^{2}Y_{j}^{\alpha
})}}-\underline{\overline{\frac{\partial }{\partial y^{j}}F((y^{0})^{2}%
\mathrm{tr}Y,(y^{0})^{4}|Y|^{2})}} \; .
\end{equation*}
We deduce from the linearity of $F$ in $\mathrm{tr}X$ and $|X|^{2}$ that $%
r^{-1}h_{j}$ admits an admissible expansion, the same holds therefore (see
corollary \ref{corollary6.6}) for $r^{-1}\underline{\xi_{j}}$. The equation
satisfied by $\underline{\nu _{j}}$ reads
\begin{equation}
r\partial _{1}\underline{\nu _{i}}+\nu _{0}\{\frac{n-1}{2}+\overline{%
\underline{r^{2}d^{hk}\delta_{hk}}}\}\underline{\nu _{i}}-2\nu _{j}r^{2}%
\underline{\overline{Y_{i}^{j}}}=h_{i}\;,  \label{8.15}
\end{equation}
\begin{equation*}
h_{i}=-\frac{1}{2}\nu _{0}r\underline{\xi _{i}}+\frac{1}{4}r^{3}\nu _{0}%
\overline{\underline{C^{hk}(\partial_{h}d_{ik}+\partial_{h}d_{ik}
-\partial_{i}d_{hk})}}\;.
\end{equation*}
An extension of Lemma 8.1 shows that $r^{-2}\underline{\nu
_{i}}$ admits an admissible expansion because it is so of
$h_{i}$.
\end{proof}

\section{The $\mathcal{C}_{0}$ constraint}

The last unknown in $\bar{g}$, only unknown in the constraint $\mathcal{C}%
_{0}$, is
\begin{equation*}
\bar{g}_{00}\equiv \underline{\bar{g}_{00}}.
\end{equation*}
The constraint $\mathcal{C}_{0}$ has a simpler expression in terms of $\bar{g%
}^{11}$. Since $\bar{g}^{11}$ is linked to $\bar{g}_{00}$ by the identity
\begin{equation*}
\bar{g}^{01}\bar{g}_{00}+\bar{g}^{11}\bar{g}_{10}+\bar{g}^{A1}\bar{g}_{A0}=0,
\end{equation*}
we have
\begin{equation}
\bar{g}_{00}\equiv -\bar{g}^{11}(\nu_{0})^{2}+\bar{g}^{AB}\nu_{B}\nu_{A}
\equiv -\bar{g}^{11}(\nu_{0})^{2}+\underline{\bar{C}^{ij}} \underline{\nu_{i}%
}\underline{\nu_{j}} .  \label{9.1}
\end{equation}

We have seen\footnote{%
See I.} that the $\mathcal{C}_{0}$ constraint can be written in vacuum as
\begin{equation}
\partial_{1}\zeta +(\kappa +\tau )\zeta +\frac{1}{2}\{\partial_{1}\bar{W}%
^{1} +(\kappa +\tau )\bar{W}^{1}+\tilde{R}-\frac{1}{2}\bar{g}%
^{AB}\xi_{A}\xi_{B} +\bar{g}^{AB}\tilde{\nabla}_{A}\xi_{B}\}=0,  \label{9.2}
\end{equation}
with
\begin{equation}
\zeta :=(\partial_{1}+\kappa +\frac{1}{2}\tau )\bar{g}^{11}+\frac{1}{2}\bar{W%
}^{1},  \label{9.3}
\end{equation}
\begin{equation}
\kappa \equiv \nu ^{0}\partial_{1}\nu_{0}-\frac{1}{2}(\bar{W}_{1}+\tau ),
\quad \bar{W}_{1}\equiv \nu_{0}\bar{W}^{0}, \quad \bar{W}^{0} \equiv \bar{W}%
^{1} \equiv -r\bar{g}^{AB}s_{AB}.  \label{9.4}
\end{equation}

\subsection{Equation for $\protect\zeta$}

In the flat case it holds that
\begin{equation*}
\nu _{0,\eta }=1,\quad \tau _{\eta }=-\bar{W}_{1,\eta }=\frac{n-1}{r},\quad
\kappa _{\eta }=0.
\end{equation*}%
The function $\zeta $ reduces to
\begin{equation}
\zeta _{\eta }:=\partial _{1}\bar{g}^{11}+\frac{1}{2}\tau _{\eta }(\bar{g}%
^{11}-1),  \label{9.5}
\end{equation}%
and the equation for $\zeta $ reads, using $\xi _{A,\eta }=0$,
\begin{equation*}
\partial _{1}\zeta _{\eta }+\frac{n-1}{r}\zeta _{\eta }+\frac{1}{2}\{\frac{%
n-1}{r^{2}}-\frac{(n-1)^{2}}{r^{2}}+\tilde{R}_{\eta }\}=0.
\end{equation*}%
That is, using the scalar curvature of the $S^{n-1}$ round sphere of radius $%
r$ which is\footnote{
See for instance \cite[p.~140]{ChBdWMII}.}
\begin{equation}
\tilde{R}_{\eta }=r^{-2}(n-2)(n-1),  \label{9.6}
\end{equation}%
the equation
\begin{equation*}
\partial _{1}\zeta _{\eta }+\frac{n-1}{r}\zeta _{\eta }=0
\end{equation*}%
with only bounded solution $\zeta _{\eta }\equiv 0$. From (\ref{9.5})
results then $\bar{g}_{\eta }^{11}\equiv 1$. We now study the general case.

We can write as follows the equation to be satisfied by $\zeta
$,
\begin{equation}
r\partial _{1}(r\zeta )+^{(r\zeta )}k\,r\zeta +^{(r\zeta )}h=0,  \label{9.7}
\end{equation}
\begin{equation}
^{(r\zeta )}k:=r(\kappa +\tau )-1\equiv r\{\nu ^{0}\partial _{1}\nu _{0}+%
\frac{1}{2}(\tau -\overline{W_{1}})\}-1,  \label{9.8}
\end{equation}%
\begin{equation*}
^{(r\zeta )}h:=\frac{r^{2}}{2}\{\partial _{1}\bar{W}^{1}+(\kappa +\tau )\bar{%
W}^{1}+\tilde{R}-\frac{1}{2}\bar{g}^{AB}\xi _{A}\xi _{B}+\bar{g}^{AB}\tilde{%
\nabla}_{A}\xi _{B}\}.
\end{equation*}%
Hence
\begin{equation}
^{(r\zeta )}h:=\frac{r^{2}}{2}\{\partial _{1}\bar{W}^{1}+\frac{1}{2}(\tau -%
\bar{W}_{1})\bar{W}^{1}+\tilde{R}+\nu ^{0}\partial _{1}\nu _{0}\bar{W}^{1}-%
\frac{1}{2}\bar{g}^{AB}\xi _{A}\xi _{B}+\bar{g}^{AB}\tilde{\nabla}_{A}\xi
_{B}\}.  \label{9.9}
\end{equation}%
We have shown that $\nu ^{0}-1$ and $r\partial _{1}\nu _{0}$, hence also $%
r\nu ^{0}\partial _{1}\nu _{0}$, admit admissible expansions of minimal
order 2, and we have seen that $r\tau $ and $r\overline{W_{1}}$ are
admissible series with terms of order zero respectively $(n-1)$ and $-(n-1)$%
. Hence $k$ is an admissible series of zero order term $n-2$,
and we have
\begin{equation*}
^{(r\zeta )}k=n-2+k_{1},
\end{equation*}%
with $k_{1}$ an admissible series of minimal order 1.

We study the terms appearing in $^{(r\zeta)}h$.

\begin{itemize}
\item
\begin{equation}
r^{2}\nu^{0}\partial_{1}\nu_{0}\bar{W}^{1}\equiv
r\nu^{0}\partial_{1}\nu_{0}\, r \bar{W}^{1}  \label{9.10}
\end{equation}
has an admissible expansion of minimal order 2 (see lemma \ref{6.4})

\item
\begin{equation}
r^{2}\bar{g}^{AB}\xi_{A}\xi_{B}\equiv r^{2}\underline{\bar{C}^{ij}}%
\underline{\xi_{i}}\underline{\xi_{j}}  \label{9.11}
\end{equation}
has an admissible expansion of order 4 because $\underline{\xi_{i}}$ has an
admissible expansion of order 1.

\item
\begin{equation*}
\bar{g}^{AB}\tilde{\nabla}_{A}\xi_{B}\equiv \bar{C}^{AB}\tilde{\nabla}%
_{A}\xi_{B},
\end{equation*}
the Christoffel symbols $\tilde{C}_{BC}^{A}$ of the Riemannian connection $%
\tilde{\nabla}$ are equal (recall that $C^{B0}=C^{B1}=C^{00}=0)$ to the
trace on $C_{O}$ of the Christoffel symbols with the same indices of the
spacetime metric $C$, hence, denoting by $^{(C)}\nabla$ the covariant
derivative in the metric $C$
\begin{equation*}
\tilde{\nabla}_{A}\xi_{B}\equiv \overline{^{(C)}\nabla_{A}\Xi_{B}}.
\end{equation*}
Since the $\Xi_{B}$ are the only non vanishing components of the vector $%
\Xi, $ and due to the form chosen for the metric $C$ we find
\begin{equation}
\bar{C}^{AB}\overline{^{(C)}\nabla_{A}\Xi_{B}}= \underline{\overline{%
C^{\alpha\beta}\; {}^{(C)}\nabla_{\alpha}\Xi_{\beta}}} ,  \label{9.12}
\end{equation}
hence the scalar $r^{2}\bar{C}^{AB}\tilde{\nabla}_{A}\xi_{B}$ has an
admissible expansion of minimal order 2.

\item We have seen that in the flat case
\begin{equation}
\bar{W}_{\eta }^{1}\equiv \eta ^{\alpha \beta }\hat{\Gamma}_{\alpha \beta
}^{1}\equiv -\frac{n-1}{r}  \label{9.13}
\end{equation}%
and
\begin{equation}
\partial _{1}\bar{W}_{\eta }^{1}+\tau \bar{W}_{\eta }^{1}\equiv \{\frac{n-1}{%
r^{2}}-\frac{(n-1)^{2}}{r^{2}}\}=-\tilde{R}_{\eta }\;.  \label{9.14}
\end{equation}%
In the general case we compute
\begin{equation*}
\partial _{1}\bar{W}^{1}+\frac{1}{2}(\tau -\bar{W}_{1})\bar{W}^{1}.
\end{equation*}%
Recall that
\begin{equation*}
\tau \equiv \frac{n-1}{r}+\frac{\overline{\mathrm{tr}X}}{r};
\end{equation*}%
set
\begin{equation*}
\bar{W}^{1}\equiv \bar{W}_{\eta }^{1}+F,\quad \text{with}\quad F:=(\bar{g}%
^{\alpha \beta }-\eta ^{\alpha \beta })\hat{\Gamma}_{\alpha \beta }^{1}.
\end{equation*}%
Using the values of the Christoffel symbols $\hat{\Gamma}_{\alpha \beta }^{1}
$ and the components of $g$ and $\eta $ in $x$ coordinates we find
\begin{eqnarray*}
F &\equiv &-(\bar{C}^{AB}-\eta ^{AB})x^{1}s_{AB}\equiv -\frac{1}{r}(\bar{C}%
^{\alpha \beta }-\eta ^{\alpha \beta })\eta _{\alpha \beta } \\
&\equiv &\frac{1}{r}(n+1-\bar{C}^{\alpha \beta }\eta _{\alpha \beta })\equiv
\frac{1}{r}(n+1-\underline{\bar{C}^{\alpha \beta }\eta _{\alpha \beta }}) \\
&\equiv &\frac{1}{r}(n-\underline{\bar{C}^{ij}\eta _{ij}})\equiv -\frac{1}{r}%
\underline{\bar{c}^{ij}}\delta _{ij},
\end{eqnarray*}%
hence
\begin{equation*}
\bar{W}^{1}\equiv -\frac{n-1}{r}-\frac{\underline{\bar{c}^{ij}}\delta _{ij}}{%
r},\quad \partial _{1}\bar{W}^{1}\equiv \frac{n-1}{r^{2}}+\frac{(\underline{%
\bar{c}^{ij}}-y^{h}\underline{\partial _{h}\bar{c}^{ij}})\delta _{ij}}{r^{2}},
\end{equation*}%
and
\begin{equation*}
\frac{1}{2}(\tau -\bar{W}^{1})\bar{W}^{1}\equiv -\{\frac{n-1}{r}+\frac{%
\overline{\mathrm{tr}X}+\underline{\bar{c}^{ij}}\delta _{ij}}{2r}\}\{\frac{%
n-1}{r}+\frac{1}{r}\underline{\bar{c}^{ij}}\delta _{ij}\}.
\end{equation*}%
Using the value of the scalar curvature $\tilde{R}_{\eta }$ of the round
sphere $S^{n-1}$ of radius $r$ we find that
\begin{equation*}
r^{2}\{\partial _{1}\bar{W}^{1}+\frac{1}{2}(\tau -\bar{W}^{1})\bar{W}%
^{1}\}\equiv -r^{2}\tilde{R}_{\eta }+\Phi,
\end{equation*}%
where $\Phi $ is an admissible expansion of minimal order 2,
\begin{equation*}
\Phi :=-y^{h}\underline{\partial _{h}\bar{c}^{ij}}\delta _{ij}-(n-2)%
\underline{\bar{c}^{ij}}\delta _{ij}-\frac{1}{2}(n-1+\underline{\bar{c}%
^{hk}\delta _{hk}})(\overline{\mathrm{tr}X}+\underline{\bar{c}^{ij}\delta
_{ij}}).
\end{equation*}

\item To compute $r^{2}\tilde{R}$ we use formulas given in I. The formulas
(10.33) and (10.37) of I for a general metric in null adapted coordinates
are
\begin{align*}
\bar{g}^{AB}\bar{R}_{AB}\equiv &\ 2(\partial_{1}+\bar{\Gamma}_{11}^{1}+\tau) %
\left[ (\partial_{1}+\bar{\Gamma}_{11}^{1}+\frac{\tau}{2}) \bar{g}^{11} +
\bar{\Gamma}^{1}\right] \\
&\ +\tilde{R}-2\bar{g}^{AB}\bar{\Gamma}_{1A}^{1}\bar{\Gamma}_{1B}^{1} -2\bar{%
g}^{AB}\tilde{\nabla}_{A}\bar{\Gamma}_{1B}^{1},
\end{align*}
\begin{equation*}
\bar{R}_{11}\equiv -\partial_{1}\tau +\bar{\Gamma}_{11}^{1}\tau
-\chi_{A}^{B}\chi_{B}^{A}
\end{equation*}
and
\begin{equation*}
\bar{S}_{01}\equiv -\frac{1}{2}\nu_{0}\bar{g}^{AB}\bar{R}_{AB}+ \bar{R}%
_{1A}\nu^{A}-\frac{1}{2}\nu_{0}\bar{g}^{11}\bar{R}_{11}.
\end{equation*}
\end{itemize}

In the case of the metric $C$ it holds that $C_{01}=1$, $C_{0A}=0$, $%
C_{00}=-1$, $C^{11}=1$ hence
\begin{equation*}
^{(C)}\Gamma_{11}^{1}\equiv ^{(C)}\Gamma_{1A}^{1}\equiv 0, \quad
^{(C)}\Gamma^{1}\equiv -\frac{1}{2}(C^{AB}\partial_{1}C_{AB}+C^{AB}%
\partial_{0}C_{AB}) ,
\end{equation*}
and the above formulas reduce to (recall that $\tilde{R}\equiv {}^{(C)}%
\tilde{R}$):
\begin{equation*}
\bar{C}^{AB}\;{}^{(C)}\bar{R}_{AB}\equiv -(\partial_{1}+\tau )\left[\tau+%
\bar{C}^{AB}\overline{\partial_{0}C_{AB}}\right] +\tilde{R} ,
\end{equation*}
\begin{equation*}
\bar{R}_{11}\equiv -\partial_{1}\tau -\chi_{A}^{B}\chi_{B}^{A},
\end{equation*}
and
\begin{equation*}
-2 \;{}^{(C)}\bar{S}_{01}\equiv \bar{g}^{AB}\bar{R}_{AB}+\bar{R}_{11},
\end{equation*}
from which we deduce
\begin{equation*}
\tilde{R}\equiv 2\partial_{1}\tau +\tau^{2}+(\partial_{1}+\tau)\bar{C}^{AB}
\overline{\partial_{0}C_{AB}}+\chi_{A}^{B}\chi_{B}^{A}-2\;{}^{(C)} \overline{%
S}_{01}.
\end{equation*}
We have
\begin{equation*}
^{(C)}\bar{S}_{01}\equiv -({}^{(C)}\underline{\bar{S}_{00}}
+r^{-1}y^{i}\;{}^{(C)}\underline{\bar{S}_{0i}}) .
\end{equation*}
Since $C$ is an analytic metric in a neighbourhood of $O$, $^{(C)}\underline{%
\bar{S}_{\alpha\beta}}$ admit admissible expansions and hence $r^2\,{}^{(C)}%
\bar{S}_{01}$ also admits an admissible expansion of minimal order 2.

Recall that
\begin{equation*}
\tau \equiv \frac{n-1}{r}+\frac{\overline{\mathrm{tr}X}}{r} \quad \text{and}
\quad |\chi|^{2}=\frac{|\bar{X}|^{2}+2\,\overline{\mathrm{tr}X}+n-1}{r^{2}} ,
\end{equation*}
hence
\begin{equation*}
r^{2}\{2\partial_{1}\tau +\tau^{2}+\chi_{A}^{B}\chi_{B}^{A}\} \equiv
(n-1)(n-2)+2(n-1)\overline{\mathrm{tr}X}+2r\partial_{1}\overline{\mathrm{tr}X%
} +(\overline{\mathrm{tr}X}\text{)}^{2}+|\bar{X}|^{2}.
\end{equation*}

Finally we remark that $\partial_{0}C_{AB}$ are the non vanishing components
of the Lie derivative of the metric $C$ with respect to the vector $m$ with $%
x$ components $m^{0}=1$, $m^{1}=m^{A}=0$ hence with $y$ components $%
\underline{m^{0}}=-1$, $\underline{m^{i}}=0$ that
\begin{equation*}
C^{AB}\partial_{0}C_{AB}\equiv C^{\alpha\beta}\mathcal{L}_{m}C_{\alpha\beta}%
\equiv \underline{C^{\alpha\beta}\mathcal{\partial}_{0}C_{\alpha\beta}} ,
\end{equation*}
hence
\begin{equation*}
\bar{C}^{AB}\overline{\partial_{0}C_{AB}}\equiv \overline{\underline{C^{ij}}}
\overline{\underline{\mathcal{\partial }_{0}C_{ij}}}\equiv \overline{%
\underline{C^{ij}}} \overline{\underline{\mathcal{\partial }_{0}c_{ij}}}
\end{equation*}
has an admissible expansion of minimal order 1 and $r^{2}(\partial_{1}+\tau)%
\bar{C}^{AB}\overline{\partial_{0}C_{AB}}$ an admissible expansion of
minimal order 2.

We have proved that
\begin{equation*}
r^{2}\tilde{R}\equiv r^{2}\tilde{R}_{\eta }+\Psi ,
\end{equation*}
where
\begin{equation*}
\Psi \equiv 2(n-1)\overline{\mathrm{tr}X} +2r\partial_{1}\overline{\mathrm{tr%
}X} +(\overline{\mathrm{tr}X})^{2} +|\bar{X}|^{2} +r^{2}(\partial_{1}+\tau)%
\overline{\underline{C^{ij}}} \overline{\underline{\mathcal{\partial}%
_{0}c_{ij}}}
\end{equation*}
has an admissible expansion of minimal order 2. Hence
\begin{equation*}
r^{2}\{\partial_{1}\bar{W}^{1}+\frac{1}{2}(\tau -\bar{W}^{1})\bar{W}^{1}+
\tilde{R}\}\equiv \Phi +\Psi .
\end{equation*}

\emph{In conclusion we have shown that $^{(r\zeta)}h$ has an admissible
expansion of minimal order 2, the same is therefore true of $r\zeta.$}

\subsection{Equation for $\protect\alpha :=\bar{g}^{11}-1$}

The definition of $\zeta $ gives for $\alpha :=\bar{g}^{11}-1$ an equation
which reads
\begin{equation}
r\partial_{1}\alpha +r(\kappa+\frac{1}{2}\tau) \alpha +\frac{r}{2}(2\kappa +
\tau +\bar{W}_{1}-2\zeta )=0.  \label{9.15}
\end{equation}

\begin{theorem}
If the basic characteristic data is induced in a neighbourhood
of $O$ by an analytic metric $C$ satisfying (4.3)-(4.4), then
the equation (\ref{9.15}) admits in this neighbourhood one and
only one \emph{solution} $\alpha $ \emph{which is an admissible
series of minimal order 2.}

This implies that $\bar{g}_{00}+1$ is also an admissible series of minimal
order 2.
\end{theorem}

\begin{proof}
Using the definition of $\kappa $ in (\ref{9.4}), the equation (\ref{9.15})
reads
\begin{equation*}
r\partial _{1}\alpha +r(\kappa +\frac{1}{2}\tau )\alpha +r(\nu ^{0}\partial
_{1}\nu _{0}-\zeta )=0,
\end{equation*}%
that is,
\begin{equation}
r\partial _{1}\alpha +{}^{(\alpha )}k\alpha +{}^{(\alpha )}h=0,  \label{9.16}
\end{equation}%
with
\begin{equation}
{}^{(\alpha )}k:=r(\nu ^{0}\partial _{1}\nu _{0}-\frac{1}{2}\bar{W}%
_{1}),\qquad {}^{(\alpha )}h:=r(\nu ^{0}\partial _{1}\nu _{0}-\zeta ).
\label{9.17}
\end{equation}%
Previous results show that this equation is of a form to which Lemma \ref%
{lemma6.5} applies with $^{(\alpha )}k_{0}=\frac{n-1}{2}$ and $^{(\alpha )}h$
of minimal order 2. It has therefore a solution $\alpha $, admissible series
of minimal order 2.

The identity (\ref{9.1}) shows the property of $\bar{g}_{00}+1$.
\end{proof}

\begin{theorem}
If in addition to the hypothesis of the previous theorem the given metric $C$
is near-round at the vertex, then $r^{-2}(\bar{g}_{00}+1)$ is an admissible
series in a neighbourhood of the vertex.
\end{theorem}

\begin{proof}
The result will follow from the proof that $^{(\alpha)}h$ given in (\ref%
{9.17}) is such that $r^{-2}{}^{(\alpha)}h$ is an admissible series. By
previous results, it remains only to prove that $r^{-2}(r\zeta)$ has an
admissible expansion, hence that it is so of $r^{-2}{}^{(r\zeta)}h$. We have
\begin{equation}
r^{-2(r\zeta)}h:=\frac{1}{2}\{ r^{-2}(\Phi+\Psi) +\nu^{0}\partial_{1}\nu_{0}%
\bar{W}^{1} -\frac{1}{2}\bar{g}^{AB}\xi_{A}\xi_{B} +\bar{g}^{AB}\tilde{\nabla%
}_{A}\xi_{B} \}.  \label{9.18}
\end{equation}
We check that, if $C$ is near-round at the vertex the various terms, studied
above, are such that the required condition is satisfied. Indeed
\begin{equation*}
\nu^{0}\partial_{1}\nu_{0}\bar{W}^{1}\equiv
\nu^{0}r^{-1}\partial_{1}\nu_{0}\, r \bar{W}^{1}
\end{equation*}
has an admissible expansion because it is so of $r^{-1}\partial_{1}\nu_{0}$
and $r\bar{W}^{1}$.
\begin{equation*}
\bar{g}^{AB}\xi_{A}\xi_{B}\equiv \underline{\bar{C}^{ij}}\underline{\xi_{i}}%
\underline{\xi_{j}}, \qquad \bar{g}^{AB}\tilde{\nabla}_{A}\xi_{B}\equiv
\underline{\overline{C^{\alpha\beta}\;{}^{(C)}\nabla_{\alpha}\Xi_{\beta }}}
\end{equation*}
have admissible expansions because $\underline{\xi_{i}}$ does.

The assumptions on $\underline{c^{ij}}$ and the identity $%
y^{h}\partial_{h}r\equiv r$ show that
\begin{equation*}
y^{h}\partial_{h}\underline{\bar{c}^{ij}\delta_{ij}} \equiv
r^{2}(y^{h}\partial_{h}\underline{\bar{d}^{ij}\delta_{ij}} +2\underline{\bar{%
d}^{ij}\delta_{ij}}),
\end{equation*}
hence $r^{-2}\Phi$, and $r^{-2}\Psi$, with $\Phi$ and $\Psi$ given above
have admissible expansions.

It remains to show the property for
\begin{equation*}
^{(C)}\bar{S}_{01}\equiv -({}^{(C)}\underline{\bar{S}_{00}}
+r^{-1}y^{i}\;{}^{(C)}\underline{\bar{S}}_{0i}).
\end{equation*}
Since $C$ is an analytic metric in a neighbourhood of $O$, $^{(C)}\underline{%
\bar{S}}_{00}$ has an admissible expansion. Denoting by $^{(C)}\underline{%
K_{ij}}\equiv -\frac{1}{2}\underline{\partial_{0}C_{ij}}$ the second
fundamental form of $C$ relative to the slicing of $R^{n+1}$ by $y^{0}=$
constant, we know that\footnote{%
See for instance~\cite[Chapter~6]{YvonneBook}.},
\begin{equation*}
^{(C)}\underline{\bar{S}_{0i}}\equiv \underline{\partial_{i}\,\mathrm{tr}%
{}^{(C)}K} -^{(C)}\underline{\bar{\nabla}_{j}(C^{jh}\,^{(C)}K_{ih})}.
\end{equation*}
We have
\begin{equation*}
^{(C)}\underline{\bar{\nabla}_{j}(C^{jh(C)}K_{ih})}\equiv \underline{C^{jh}}%
( \underline{\partial_{j}\,^{(C)}K_{ih} -C_{ji}^{k}{}^{(C)}K_{ik}
-C_{jh}^{k(C)}K_{jk}}) .
\end{equation*}
The functions $^{(C)}\underline{K_{ij}}$ are analytic and the Christoffel
symbols $\underline{C_{ji}^{k}}$ are products by $y^{0}$ of analytic
functions, while elementary computations give
\begin{equation*}
-2y^{i}\underline{C^{jh}}\underline{\partial_{j}\,^{(C)}K_{ih}}\equiv y^{i}%
\underline{C^{jh}}\underline{\partial_{j}\partial_{0}C_{ih}}\equiv
\underline{C^{jh}}\{\underline{\partial_{j}\partial_{0}(y^{i}c_{ih})}-%
\underline{\partial_{0}c_{jh}}\},
\end{equation*}
hence using $y^{i}\underline{c_{ih}}=0$ and $\delta^{jh}\underline{c_{jh}}%
=(y^{0})^{2}Z$. We deduce from these results that $r^{-1}y^{i}{}^{(C)}%
\underline{\bar{S}}_{0i}$ also has an admissible expansion.

The proof is complete.
\end{proof}

\section{Conclusions}

\subsection{An existence theorem}

We have shown that when the metric $C$ given by (\ref{4.3}) which induces on
$C_{O}$ the basic characteristic data $\tilde{g}$ (i.e. $\bar{g}_{AB}\equiv
\bar{C}_{AB}\equiv \tilde{C}_{AB})$ is analytic, then the functions $\nu
_{0} $, $\underline{\nu _{i}}$, $\underline{\bar{g}_{00}}$ have admissible
expansions. We have shown that if moreover $C$ is near-round at the vertex
(definition \ref{7.2}), then the functions $r^{-2}(\nu _{0}-1)$, $r^{-1}%
\underline{\nu _{i}}$, and $r^{-2}(\underline{\bar{g}_{00}}+1)$
have also admissible expansions. These results imply the
following theorem (the notations are those of section 2):

\begin{theorem}
If the metric $C$ given by (\ref{4.3})-(\ref{4.4}) which induces the basic
characteristic data on the cone $C_{O}^{T}$ is smooth everywhere, and
moreover analytic and near-round in a neighbourhood of the vertex, then
there exists a number $T_{0}>0$ such that the wave-gauge reduced vacuum
Einstein equations with characteristic initial determined by $C$ and the
solution of the wave-map gauge constraints have a solution in $Y_{O}^{T_{0}}$
which induces on $C_{O}^{T_{0}}$ the same quadratic form as $C$.
\end{theorem}

\begin{proof}
It results from the formulae
\begin{equation}
\underline{\bar{g}_{00}}\equiv \bar{g}_{00},\quad \underline{\bar{g}_{0i}}%
\equiv -\{\bar{g}_{00}+\nu _{0})\}r^{-1}y^{i}-\underline{\nu _{i}},\quad
\text{with}\quad y^{i}\underline{\nu _{i}}\equiv 0\;,  \label{10.1}
\end{equation}
\begin{equation}
\quad \underline{\bar{g}_{ij}}-\delta _{ij}\equiv \{\bar{g}_{00}+1+2(\nu
_{0}-1)\}r^{-2}y^{i}y^{j}+r^{-1}(y^{i}\underline{\nu _{j}}+y^{j}\underline{%
\nu _{i}})+\underline{c_{ij}},  \label{10.2}
\end{equation}
and the theorems of previous sections that $\underline{\bar{g}_{00}}+1$, $%
\underline{\bar{g}}_{0i}$, and $\underline{\bar{g}}_{ij}-\delta _{ij}$ have
admissible expansions, hence are the trace on $C_{O}$ of analytic functions.
We apply the Cagnac-Dossa theorem.
\end{proof}

This theorem and the results of I lead then to the following sufficient
condidions for the existence of a solution of the full Einstein equations.

\begin{theorem}
If the metric $C$ given by (\ref{4.3}) which induces the basic
characteristic data on the cone $C_{O}$ is smooth, and analytic and
near-round at the vertex, there exists a vacuum Einsteinian spacetime $%
(Y_{O}^{T_{0}},g)$ which induces on $C_{O}^{T_{0}}$ the same quadratic form
as $C$. The solution is locally geometrically unique.
\end{theorem}

\noindent{\textsc{Acknowledgements:} JMM was supported by the French ANR
grant BLAN07-1\_201699 entitled ``LISA Science'', and also in part by the
Spanish MICINN project FIS2009-11893. PTC was supported in part by the
Polish Ministry of Science and Higher Education grant Nr N N201 372736. }

\providecommand{\bysame}{\leavevmode\hbox
to3em{\hrulefill}\thinspace} \providecommand{\MR}{\relax\ifhmode\unskip%
\space\fi MR }
\providecommand{\MRhref}[2]{  \href{http://www.ams.org/mathscinet-getitem?mr=#1}{#2}
} \providecommand{\href}[2]{#2}

\end{document}